%% file: main.tex
\documentclass[12pt]{article}

\RequirePackage[OT1]{fontenc}
\RequirePackage{graphicx}
\RequirePackage{amsthm}
\RequirePackage[cmex10]{amsmath}
\RequirePackage{natbib}
\RequirePackage[colorlinks,citecolor=black,urlcolor=black,linkcolor=black]{hyperref}
 
\newtheorem{notation}{Notation}
\newtheorem{lemma}{Lemma}
\newtheorem{fact}{Fact}

\newtheorem{remark}{Remark}
\newtheorem*{claim*}{Claim}
\newtheorem{theorem}{Theorem}

\newtheorem{example}{Example}
\newtheorem{corollary}{Corollary}
\newtheorem{assumption}{Assumption}
\newtheorem{proposition}{Proposition}
\newtheorem{definition}{Definition}

\usepackage[top=1.18in, bottom=1.18in, left=1.2in, right=1.2in]{geometry}
 
\newtheorem{innercustomthm}{Proposition}
\newenvironment{propz}[1]
  {\renewcommand\theinnercustomthm{#1}\innercustomthm}
  {\endinnercustomthm}
 
 \usepackage{epsfig,amsfonts,setspace,dsfont}
\usepackage{amsthm,array,geometry, wasysym,enumerate}
\usepackage{sgame}
\usepackage{subcaption}
\usepackage{comment}
\usepackage[normalem]{ulem}
\usepackage{amsmath,amssymb,amsthm}
\usepackage{setspace,graphicx}
\usepackage{lmodern}
\usepackage{slantsc}%
\usepackage[cmtip,all]{xy}
\usepackage{accents}
\usepackage[cmyk]{xcolor}
\usepackage{tikz}
\usepackage{ifthen}
\usepackage{xfrac}
\usepackage{txfonts}
\usepackage{epstopdf}
\usepackage{times} 
\usepackage{titlesec}

\usetikzlibrary{calc,arrows,automata,shapes.misc,shapes.arrows,chains,matrix,positioning,scopes,decorations.pathmorphing,shadows}

\newcommand{\citeapos}[1]{\citeauthor{#1}'s \citeyearpar{#1}}

\DeclareMathOperator{\co}{co}
\DeclareMathOperator{\gr}{gr}
\newcommand{\marg}{{\rm{marg}}}
\newcommand{\dd}{{\rm{d}}}
\newcommand{\ddd}{\ \dd}
\DeclareMathOperator{\supp}{supp}
\DeclareMathOperator{\ca}{ca}
\newcommand{\RP}{\mathcal I}
\newcommand{\RPP}{\RP(\prior)}
\newcommand{\ext}{{\rm{ext}}}
\newcommand{\DT}{{\Delta\Theta}}

\newcommand{\real}{\mathbb{R}}
\newcommand{\longsquiggly}{\xymatrix{{}\ar@{~>}[r]&{}}}

\newcommand{\rstrat}{\alpha}
\newcommand{\beliefr}{\pi}

\newcommand{\beliefm}{\beliefr}

\newcommand{\cstrat}{\psi}
\newcommand{\cstrats}{\Psi}
\usepackage{bm}

\newcommand{\us}{u_{S}}
\newcommand{\ur}{u_{R}}

\newcommand{\infpay}{s_i}

\newcommand{\condpay}{\phi}

\newcommand{\goodm}{g}
\newcommand{\badm}{b}

\newcommand{\poolp}{k}
\newcommand{\cred}{\chi}
\newcommand{\cstratr}{{\tilde\cstrat}}
\newcommand{\sstrat}{\sigma}
\newcommand{\cstratrs}{\Psi}
\newcommand{\sstrats}{\Sigma}

\newcommand{\prior}{{\mu_0}}

\makeatletter
  \renewcommand\@seccntformat[1]{\csname the#1\endcsname.{\hskip.7em\relax}} 
\makeatother

\usetikzlibrary{calc,arrows,automata,shapes.misc,shapes.arrows,chains,matrix,positioning,scopes,decorations.pathmorphing,shadows,arrows.meta}
\definecolor{orange}{HTML}{F4A261} 
\definecolor{blue}{HTML}{4D8EA8} 
\definecolor{green}{HTML}{2faf6d} 
\definecolor{red}{HTML}{E45C3A}
\definecolor{yellow}{HTML}{F9C74F}
\definecolor{purple}{HTML}{560bad}

\title{{\bf 
Perfect Bayesian Persuasion\footnote{Lipnowski and Ravid acknowledge support from the National Science Foundation (grant SES-1730168). Lipnowski acknowledges support from the Columbia Program for Economic Research. We would like to thank Laura Doval, Christoph Kuzmics, and Stephen Morris for helpful comments. Zikai Xu provided excellent research assistance.}
}}

\author{
\begin{minipage}{0.3\textwidth}\centering  
Elliot Lipnowski\footnote{\texttt{e.lipnowski@columbia.edu}, \texttt{dravid@uchicago.edu}, \texttt{dshishkin@ucsd.edu}} \\ \centering \it \small Columbia University
\end{minipage}                  
\begin{minipage}{0.3\textwidth}\centering 
Doron Ravid  \\ \centering \it \small University of Chicago 
\end{minipage} 
\begin{minipage}{0.3\textwidth}\centering 
Denis Shishkin
\\ \centering \it \small UC San Diego 
\end{minipage}                  
}

\newcommand{\argmax}{\mathrm{argmax}}

\date{\vspace{0.8cm} \today}
\begin{document}

\maketitle
\begin{abstract}

\noindent A sender commits to an experiment to persuade a receiver.  Accounting for the sender's experiment-choice incentives, and not presupposing a receiver tie-breaking rule when indifferent, we characterize when the sender's equilibrium payoff is unique and so coincides with her ``Bayesian persuasion'' value.  
A sufficient condition in finite models is that every action which is receiver-optimal at some belief is uniquely optimal at some other belief---a generic property. 
We similarly show the equilibrium sender payoff is typically unique in ordered models.
In an extension, we show uniqueness generates robustness to imperfect sender commitment.
\\   

\vspace{.2cm}

\end{abstract}
\newpage 

\begin{spacing}{1}
\onehalfspacing


\section{The Setup}

In recent years, many studies have explored variations on the Bayesian Persuasion model of \cite{Kamenica2011}, hereafter KG \citep[see][for review]{bergemann2019information,Kamenica2019}. In KG's model, a receiver (R, he) must choose an action, but a sender (S, she) controls R's available information about a payoff-relevant state. When analyzing this model and its variations, researchers usually assume S-favorable tie-breaking. Whereas this assumption is without loss when evaluating the highest payoff S can get across all information structures, making this assumption when using the model for prediction or design is tantamount to selecting S's favorite equilibrium. Consequently, it is unclear whether the model's conclusions remain valid under other equilibrium selection criteria. In this paper, we address this question by asking when tie-breaking assumptions are irrelevant. More specifically, we ask: when is S's equilibrium payoff unique?

To answer this question, we study KG's model without imposing S favorable tie breaking. Thus, the game begins with S publicly committing to a Blackwell experiment about the payoff-relevant state $\theta \in \Theta$, whose prior distribution is $\mu_0 \in \DT$. Formally, S chooses a measurable function $\cstrat\colon\Theta\to\Delta M$ for a given space $M$ of messages. Then, after observing both $\cstrat$ and a realized message $m\in M$, the receiver chooses an action $a$ from a set of feasible actions $A$. Each player $i\in\{S,R\}$ seeks to maximize the expectation of an objective $u_i(a,\theta)$.\footnote{We maintain the technical assumptions that both $A$ and $\Theta$ are nontrivial compact metrizable spaces; that the objectives $u_S,u_R\colon\, A\times\Theta\to\real$ are continuous; that $M$ is a Polish space; and either that $M$ is uncountable or that $\Theta$ is finite and $|M|\geq|\Theta|$. All topological spaces are viewed as measurable spaces with their Borel sigma-algebra.}
KG analyze this game using the S-favorite equilibrium concept. We diverge from this assumption by considering all equilibria of the game, and ask when all these equilibria give S the same payoff. 

We defer a formal definition of equilibrium to the following section, but a brief overview of KG's analysis is in order. They adopt a belief-based approach, casting S's optimization problem as one of directly choosing $p \in \Delta\DT$, the ex-ante distribution of R's posterior belief $\mu$ concerning the state. Because R is Bayesian (and S's experiment choice is made in ignorance of the state), it must be that $p$ belongs to $\RP(\mu_0)$, the set of belief distributions with barycenter equal to the prior; KG term this condition Bayes plausibility. We refer to elements of $\RP(\mu_0)$ as \textbf{information policies}. But what payoff does S derive from a given Bayes-plausible belief distribution? By R rationality, R will choose from his best responses $A^*_R(\mu) \subseteq A$ whenever her posterior belief is $\mu$.  S's expected payoff from such an action $a\in A^*_R(\mu)$ is then equal to $\int u_S(a,\cdot) \ddd\mu$. Given KG's focus on S-optimal equilibrium, they can assume without loss that R always breaks any indifferences in S's favor. Thus, KG can summarize S's payoff from inducing belief $\mu$ as $v(\mu)\coloneqq\max_{a\in A^*_R(\mu)}\int u_S(a,\cdot) \ddd\mu$. We call $v$ the \emph{value function}. Hence, S's best equilibrium value is given by $\hat v(\mu_0)=\max_{p\in\RP(\mu_0)} \int v \ddd p$.

But what happens if R may choose best responses that differ from those S would prefer? If, in the worst case, R always chooses S's least favorite of his best responses, then from inducing R belief $\mu$, S can only expect a payoff of $w(\mu)\coloneqq\min_{a\in A^*_R(\mu)}\int u_S(a,\cdot) \ddd\mu$. Accordingly, S would have a profitable deviation if her payoff were ever strictly below $$\hat w(\mu_0)=\sup_{p\in\RP(\mu_0)} \int w \ddd p.$$ As \cite{wu2020essays} has shown, this payoff lower bound is the only additional constraint imposed by S's experiment-choice incentives. Using this result, we go on to fully characterize when S has a unique equilibrium payoff, and to provide meaningful sufficient conditions for the same.

One can see the issues that arise without S-favorable tie-breaking by modifying KG's judge example. In their example, a binary state in $\Theta = \{0,1\}$ summarizes whether a defendant is innocent ($\theta = 0$) or guilty ($\theta = 1$). S is a prosecutor, who decides what information to generate about this state, whereas R is a judge. In KG's example, the judge makes a binary decision as to whether to acquit ($a=0$) or convict ($a=1$) the defendant, and S always wants the defendant convicted, whereas R wants to make the ``just'' choice, getting a utility of $1$ if her choice matches the state. Suppose we modify the example so that R has two convict actions, $a=1$ and $a=-1$, which differ in the defendant's sentence. For example, in a murder trial, $a=1$ may represent sentencing the defendant to a life in prison, whereas $a=-1$ represents giving the defendant the death penalty. Suppose the (American) judge views both sentences as equivalent to each other, but the prosecutor views the death penalty as immoral, and so prefers that the judge acquits the defendant rather than give him the death penalty. For concreteness, suppose S's preferences are given by $u_S(a,\theta) = a$, whereas R gets a payoff of $1$ if her sentencing choice matches the state (i.e., if $|a|=\theta$), and $0$ otherwise. For our example, suppose the prior probability the defendant is guilty is $\mu_0=0.25$.

\input{modified-judge-pic.tex}

Figure \ref{fig:judge-example-penalty} depicts the sender's upper and lower value functions, $v$ and $w$, as well as their ``concavifications'', $\hat{v}$ and $\hat{w}$, for this example. As usual, the figure summarizes R's belief by the probability it assigns to $\theta = 1$. Notice that $\hat{v} \neq \hat{w}$, and so S's value would, in general, depend on the tie-breaking rule. In fact, below we explain how different tie-breaking rules lead to radically different conclusions about S's equilibrium value and the information she chooses to release to R. 

Suppose first that R always breaks ties in S's favor. As KG explain, in this case S's expected payoff conditional on R having a belief $\mu$ is equal to $v(\mu)$, and so S's payoff from any given information policy $p$ is given by $\int v \ddd p$. S's value from the best information policy under S-favorable equilibrium selection is then given by $\hat{v}(\mu_0)$. In this example, S's equilibrium payoff under S-favorable tie-breaking is $\hat{v}(\mu_0)=0.5$. The equilibrium information policy is then $p^*=0.5 \delta_0 + 0.5\delta_{0.5}$. This is the solution prescribed by the majority of the information design literature. Note this payoff is strictly higher than $0$, which is  the utility S would get if she revealed no information to R. 

What if R were to break ties \emph{against} S? In this case, S's payoff from inducing a belief of $\mu$ would be given by $w(\mu)$, and so S's payoffs from $p^*$ would be equal to 
\[
0.5 w(\delta_0) + 0.5 w(\delta_{0.5}) = -0.5.
\]
Notice this payoff is strictly below S's no information payoff, which equals $0$ under both S-favorable and S-adverserial tie breaking. In fact, one can show that, under S-adversarial tie-breaking, no information gives S her best equilibrium payoff. By contrast, if R breaks ties against S, $p^*$ gives S the lowest payoff S can attain under \emph{any} information policy, subject to R's incentive constraints.

As the above example demonstrates, multiplicity in S's equilibrium value can result in predictions that are highly sensitive to R's tie-breaking rule. Our current paper provides tools for identifying situations in which such sensitivity is not present by characterizing the environments in which S's equilibrium value is unique. 

We also apply our results to a model in which there are doubts about S's ability to commit to her experiment. Such doubts make the equilibrium selection question salient, since they suggest S may lack the bargaining power required to dictate which equilibrium is played. Section~\ref{sec:pwi} introduces a model in which S's commitment power is limited. Specifically, we add a probability that S can covertly manipulate the experiment's results ex-post. We characterize S's equilibrium payoff set as this probability vanishes, and apply our uniqueness results to study robustness in that setting.

\paragraph{Related literature}
We contribute to the Bayesian persuasion literature \citep{Aumann1995,Kamenica2011,Kamenica2019}, which studies sender-receiver games in which a sender commits to an information-transmission strategy. Our main goal is to understand when KG's sender payoff characterization still applies---taking her experiment-choice incentives into account, and allowing the receiver to choose actions in a manner unfavorable to the sender when he is indifferent. Our extension to the case in which the sender's commitment is imperfect extends the analysis of \cite{LRS1} in the same manner.\footnote{
See also \cite{Min:2016wv}, which develops a generalization of the limited-commitment model and shows some credibility is better than no credibility in a leading example; and  \cite{Frechette:2017ug}, which studies communication outcomes in a laboratory experiment.}

Several papers introduce robustness considerations into information design problems \citep[e.g.,][]{kosterina2022persuasion,dworczak2022preparing}.
Within this literature, most relevant are papers that study optimal information provision in multi-agent settings under planner-worst equilibrium selection \citep[e.g.][]{mathevet2020information,moriya2020asymmetric,oyama2020generalized,ziegler2020adversarial,halac2022addressing,morris2022implementation,li2023global,inostroza2023adversarial}.
Our work provides conditions under which such a selection rule is irrelevant for the single-agent case.

Our work is related to multiple distinct strands of the Bayesian persuasion literature that explicitly account for sender incentives in choosing an experiment. The first is the literature on competition in persuasion, in which multiple senders design flexible information either simultaneously \citep{gentzkow2016competition,gentzkow2017bayesian,au2020competitive,ravindran2020competing} or sequentially \citep{li2018sequential,wu2020essays} and must make individually rational experiment choices in equilibrium. The most related paper from this set is \cite{wu2020essays}, which essentially proves Proposition~\ref{perprior}.%
\footnote{A result complementary to our Proposition~\ref{generic} was proven by 
\cite{li2018sequential}. They study persuasion by multiple senders moving sequentially, assuming the receiver breaks ties in favor of the sender who moves last. Specialized to the single sender case, their Proposition 3 implies that, \emph{assuming} sender-favorable tie-breaking, all equilibria result in the same state-contingent action distribution. Combined with our Proposition~\ref{generic}, \citeapos{li2018sequential} result implies that in generic finite environments, one gets behavioral uniqueness (in addition to payoff uniqueness).} 
The second strand studies how experiment constraints shape chosen information \citep[e.g.,][]{ichihashi2018limiting,PerezRichet:2017va}, whereas the third strand studies signaling and informed principal problems for information design settings in which an experiment choice can reveal private information \citep[e.g.,][]{perez2014interim, Hedlund:2017ie,Alonso:2018he,koessler2021information}.\footnote{In a sense, the literature on verifiable disclosure can be seen as a combination of the second and third strand, with sender incentives being a key object of study and verifiable information being limited in a type-dependent manner.} Finally, beyond the Bayesian persuasion literature, work on sequential mechanism design under limited commitment \citep[e.g.,][]{Skreta2006,doval2020mechanism} accounts for a principal's incentives while her future beliefs play a prominent role in the analysis.

\section{The Equilibrium Payoff Set}\label{sec:fullset}

We now formally define an equilibrium concept for the persuasion game. Note two features of the definition. First, while R must respond optimally to his belief, we make no direct assumption concerning which best response he chooses when indifferent.\footnote{Still, as we shall see, it will often be the case that \emph{mutual} best response requires that R break indifferences in S's favor on the path of play, just as a recipient of a zero offer accepts the offer in the unique subgame-perfect equilibrium of the ultimatum game.} Second we explicitly include an optimality condition for S at the experiment-choice stage. Note that the latter condition would have no bite under S-favorable tie-breaking, and so is rarely included in the literature. 

\begin{definition}\label{def:eqm}
Let $\cstratrs$ denote the set of all measurable functions $\cstratr\colon\Theta\to\Delta M$ (a.k.a. experiments).  
A sender strategy is an experiment $\cstrat\in\cstrats$; 
a receiver strategy is a function $\rstrat\colon M \times \cstratrs\to\Delta A$ that is measurable in its first argument;
and a receiver belief map is a function $\beliefm\colon M \times \cstratrs\to\DT$ that is measurable in its first argument.  A \textbf{(perfect Bayesian) equilibrium} is a triple of such maps $\langle\cstrat,\rstrat,\beliefm\rangle$ such that \begin{enumerate}
\item The sender's choice satisfies 
$$\cstrat
\in \argmax_{\cstratr\in\cstratrs} \int_\Theta \int_M \int_A u_S(a,\theta) \ddd\rstrat(a|\cstratr,m) \ddd\cstratr(m|\theta)\ddd\mu_0(\theta);$$
\item Every $\cstratr\in\cstratrs$ and $m\in M$ have $$\rstrat\left(\argmax_{a\in A} \int_\Theta u_R(a,\theta) \ddd\beliefm(\theta|m,\cstratr) \ \middle|\  m,\cstratr \right)=1;$$
\item Every $\cstratr\in\cstratrs$, Borel $\hat M\subseteq M$, and Borel $\hat\Theta \subseteq \Theta$ have $$\int_\Theta \int_{\hat M} \beliefm(\hat\Theta|m,\cstratr) \ddd\cstratr(m|\theta) \ddd\mu_0(\theta) = \int_{\hat\Theta} \cstratr(\hat M|\theta) \ddd\mu_0(\theta).$$
\end{enumerate}
In such a case, we say the induced \textbf{equilibrium sender payoff} is
$$
\int_\Theta \int_M \int_A u_S(a,\theta) \ddd\rstrat\left(a\middle|m,\cstrat
\right) \ddd\cstrat(m|
\theta)\ddd\mu_0(\theta)
.$$
\end{definition}

The interpretation is as follows.  First, S publicly chooses an experiment $\cstratr\in\cstratrs$.
\footnote{One could easily extend the model to allow S to mix over experiment choice.  Doing so would entail added notational burden but would have no effect on the resulting S payoff set because the experiment choice is public, not informed by private information, and not simultaneous to any other decisions.
} The experiment then produces a message $m\in M$ that R observes.  Then, R updates his beliefs according to the message and the chosen experiment, and chooses an action $a\in A$.  We require that S 
only choose experiments that maximize her expected payoffs, that R (having seen the realized experiment and message) only choose actions that maximize his expected payoffs with respect to his belief about the payoff state, and that R's beliefs conform to Bayesian updating.\footnote{Moreover, we assume that S cannot signal what she does not know. 
Indeed, our Bayesian condition implies that every $\cstratr\in\cstratrs$ and Borel $\hat\Theta \subseteq \Theta$ have $\int_\Theta \int_{M} \beliefm(\hat\Theta|m,\cstratr) \ddd\cstratr(m|\theta) \ddd\mu_0(\theta) = \mu_0(\hat\Theta)$, 
so that the chosen experiment alone does not cause belief updating by R about the payoff state.}

In what follows, we document the set of attainable equilibrium S payoffs, with a particular focus on understanding when it is unique. 

\begin{remark}
Although our focus is on equilibrium S payoffs rather than behavior, our results have natural implications for behavior as well. In particular, when S's equilibrium payoff is unique, our results imply that R breaks indifferences in S's favor with probability 1 on path in every equilibrium. Hence, in this case, the results of KG \cite[and many subsequent papers surveyed in][]{Kamenica2019} are robust to allowing arbitrary tie-breaking for R and accounting for S's experiment-choice incentives.
\end{remark}

\subsection{Characterizing equilibrium payoffs}

We begin by stating a characterization of the equilibrium S payoff set as a function of the parameters of our game.  This set is a compact interval, with highest value equal to KG's commitment solution, and lowest value equal to the supremum value S can guarantee when R breaks her indifferences adversarially.  In the special case in which the state space is finite, this result is exactly Proposition 1 from \cite{wu2020essays}. Although no substantive new arguments are required for the general case, we include a proof for the sake of completeness.

\begin{propz}{0}[Payoff set \citep{wu2020essays}]\label{perprior}
The set of equilibrium S payoffs is $[\hat w(\mu_0), \hat v(\mu_0)]$.
\end{propz}

Necessity is essentially immediate, and the proof of sufficiency is constructive.  By degrading information from an S-optimal (under favorable tie-breaking) experiment and allowing for R to mix among optimal choices in the degraded experiment, one can find an experiment for S to choose and R best response to target any payoff in the given interval. Then, having R break indifference adversarially to S following off-path experiment choices ensures that this experiment choice is indeed optimal for S.

\begin{remark}\label{rem: general value function results}
It is apparent that Proposition~\ref{perprior} depends only on the value correspondence $V=[w,v]$, and moreover (as is clear from our proof) the only substantive property required of the environment is that the attainable S payoffs from R responding optimally to a given belief be convex.\footnote{Our analysis also uses the fact that $V$ is nonempty-compact-valued and upper hemicontinuous, and that the set of optimal R choices is a weakly measurable correspondence of his belief.} In addition to making Proposition~\ref{perprior} more tractable to apply, this feature also expands its applicability beyond the basic model we have considered.  For example, the proposition can be applied to settings in which a receiver is subject to independent private payoff shocks.  Additionally, the proposition applies to public persuasion of a set of agents who play a game, so long as the set of induced payoffs for the sender is convex for every public belief.  The latter condition holds, for instance, if the receivers observe a rich public randomization device after the experiment choice (but before their gameplay).
\end{remark}

\begin{remark}
In light of Proposition~\ref{perprior}, S has a unique equilibrium payoff if and only if $v$ and $w$ have the same ``concavification,'' and so too does any function $z$ with $w\leq z\leq v$. We can therefore frame our analysis as follows: When is the concavification of a sender's value function (evaluated at the prior) invariant to R's tie-breaking rule?
\end{remark}

\section{Equilibrium Payoff Uniqueness}\label{sec:unique}

In this section we ask, when does S have a unique equilibrium payoff? Whenever she does, the traditional analysis that focuses on S-optimal equilibrium (and so assumes S-favorable tie-breaking by R) is essentially without loss.

As a starting observation, because uniqueness follows directly from Proposition~\ref{perprior} whenever $v=w$, a sufficient condition for S to have a unique equilibrium payoff is immediate.
\begin{corollary}[No relevant ties]
S has a unique equilibrium payoff if, at any belief, S is indifferent between all of R's best responses.
\end{corollary}
Although restrictive, the above condition nevertheless captures many cases of interest. For example, if the action space is a convex subset of some linear space with R's payoff being strictly concave in his action \citep[e.g.,][]{Crawford1982,Chakraborty2010}, then he has a unique best response to every belief, and so the corollary applies.

The following result gives an alternative sufficient condition for S to attain her Bayesian persuasion value in all equilibria.  It says such uniqueness holds if information can serve as a stand-in for favorable tie-breaking at all relevant posterior beliefs. To state the result, the following definition is useful.

\begin{definition}\label{ubr} 
Say a set $D\subseteq\DT$ is \textbf{persuasion sufficient} if 
\[
\hat v(\mu_0)=
\sup \left\{\int v\ddd p \colon\, p\in\RP(\mu_0) \ \text{s.t.} \ p(D')=1\text{ for some Borel } D'\subseteq D \right\}.
\]
\end{definition}
In words, a set of beliefs $D$ is persuasion sufficient if the KG payoff can be approximated with policies supported within $D$ under S-favorable tie-breaking.

As the following proposition says, if (on the relevant set of beliefs) information can be used to replace favorable selection, then S guarantees her KG payoff.

\begin{proposition}[Information as selection]\label{infoselect}
S has a unique equilibrium payoff if $\hat w|_D\geq v|_D$ for some persuasion-sufficient $D\subseteq\DT$.\footnote{For a function $f\colon\Delta \Theta \rightarrow \mathbb{R}$, we let $f|_{D}$ be the restriction of $f$ to $D$; that is, the function from $D$ to $\mathbb{R}$ defined via $\mu \mapsto f(\mu)$.}
\end{proposition}

To prove the proposition, we begin with an information policy supported on $D$ that nearly attains the KG payoff under S-favorable selection. For each realized posterior $\mu\in D$, the sender then provides additional information---further splitting the belief $\mu$ into some belief distribution centered on $\mu$---to ensure a payoff of nearly $v(\mu)$ even under S-adversarial selection. That $\hat w\geq v$ on $D$ ensures one can find such a splitting for each $\mu$, and a standard measurable selection result says one can do so measurably.

Following directly from the above proposition, the next corollary exactly characterizes when equilibrium S payoff uniqueness holds independent of the prior: such uniqueness is equivalent to information always replicating favorable tie-breaking.

\begin{corollary}[Global uniqueness]\label{allprior}
S has a unique equilibrium payoff for every prior (holding other parameters fixed) if and only if $\hat w\geq v$.
\end{corollary}

\subsection{Sufficient conditions for uniqueness}

The above conditions for payoff uniqueness were expressed in terms of the derived objects $v$, $w$, $\hat v$, and $\hat w$. These objects can be difficult to calculate. In this section, we develop alternative sufficient conditions that ensure S has a unique equilibrium payoff.

Let $A^U\subseteq A$ denote the set of {potentially unique best response actions}, that is, $a\in A$ such that some $\mu\in\Delta[\supp\prior]$ has $A^*_R(\mu)=\{a\}$. For any belief $\bar\mu\in\DT$, let $A^U(\bar\mu)$ denote the set of {potentially unique best response actions at $\bar\mu$}, that is, $a\in A$ such that some $\mu\in\DT$ has both $A^*_R(\mu)=\{a\}$ and $\varepsilon\mu\leq\bar\mu$ for some $\varepsilon>0$.\footnote{When we refer to inequalities for measures, we interpret the inequality pointwise. For instance, $\varepsilon\mu\leq\bar\mu$ means $\varepsilon\mu(\hat\Theta)\leq\bar\mu(\hat\Theta)$ for every measurable $\hat\Theta\subseteq\Theta$.}
\begin{definition}\label{ubr} 
Given $\mu\in\DT$, say \textbf{the potentially unique best response (PUBR)  property holds at $\mu$} if\footnote{Note the left-hand side is just $v(\mu)$.} $$\max_{a\in A^*_R(\mu)}\int u_S(a,\cdot)\ddd\mu =\sup_{a\in A^*_R(\mu)\cap A^U}\int u_S(a,\cdot)\ddd\mu;$$
and say \textbf{the strong PUBR property holds at $\mu$} if $$\max_{a\in A^*_R(\mu)}\int u_S(a,\cdot)\ddd\mu =\sup_{a\in A^*_R(\mu)\cap A^U(\mu)}\int u_S(a,\cdot)\ddd\mu.$$
Given a set $D\subseteq\DT$, say \textbf{the (strong) PUBR property holds on $D$} if it holds at every $\mu\in D$.
\end{definition}
The above property says that S can obtain a value arbitrarily close to her payoff under favorable-selection with the receiver only using actions that are a unique best response to some belief---and the strong version further requires that the latter belief be boundedly absolutely continuous with respect to the relevant posterior belief.\footnote{Note, the PUBR property is strictly stronger than the requirement that R has no duplicate actions. The latter condition is insufficient for ensuring uniqueness, as witnessed by $A=\Theta=\{0,1\}$, $\mu_0=\tfrac12$, $u_S(a,\theta)= - a$, and $u_R(a,\theta)=a\theta$.
The condition is reminiscent of the refined best reply correspondence that \cite{balkenborg2013refined} and \cite{balkenborg2015refined} study in finite games of complete information.
}

The next result shows various versions of the PUBR property are sufficient to guarantee equilibrium selection. 

\begin{theorem}[The PUBR theorem\footnote{Cf. \cite{milne1926winnie}.}]\label{ubrsuff}
S has a unique equilibrium payoff if either of the following two conditions holds for some persuasion-sufficient $D\subseteq\DT$: \begin{enumerate}[(i)]
    \item The strong PUBR property holds on $D$.
    \item The PUBR property holds on $D$, and either $A$ or $\Theta$ is finite.
\end{enumerate}
\end{theorem}
In particular, for full-support priors and finite states or finite actions, the theorem implies S has a unique equilibrium payoff if every action is the unique best response to some belief.

For intuition, consider the case in which the state and action spaces are both finite and the prior has full support, and let $p$ be any Bayes-plausible belief distribution with finite support $\tilde D$ satisfying the PUBR property. For any supported belief $\mu$, the PUBR property implies some action $a$ is both (A) an S-preferred R best response to belief $\mu$ and (B) a unique R best response to some alternative belief $\mu'_{\mu}$. For any $\lambda\in(0,1)$, we can then define the belief $\mu'_{\mu,\lambda}=\lambda \mu'_{\mu} + (1-\lambda)\mu$, which also has $a$ as a unique R best response by the linearity property of expected utility. Define now the belief distribution $p_\lambda$ which modifies $p$ by replacing each supported $\mu$ with $\mu'_{\mu,\lambda}$. This belief distribution is Bayes-plausible for the alternative prior $\mu_\lambda=(1-\lambda)\prior+\lambda\sum_{\mu}p(\mu)\mu'_\mu$, which converges to $\prior$ as $\lambda\to0$. Moreover, because R has a unique best response to every $p_\lambda$-supported message, it follows that $\int w\ddd p_\lambda$ converges to $\int v \ddd p$ as $\lambda\to 0$. Therefore, every limit point of $\hat w(\mu_\lambda)$ as $\lambda\to0$ is at least $\int v \ddd p$. Finally, the function $\hat w$ is concave---for the exact same reason that the optimal value $\hat v$ is---and so is continuous at the relatively interior point $\prior$ of its domain. Hence, $\hat w(\prior)\geq\int v \ddd p$, delivering the result.

There are two hurdles to generalizing the above argument to more general action and state spaces. The first hurdle is that for some beliefs, there may be no action that satisfies both (A) and (B) above. To circumvent this issue, we replace (B) with the requirement that the action can be approximated by actions that are unique R-best responses at some alternative beliefs. The second hurdle is that, with infinite states, $\DT$ has an empty interior when viewed as a subset of the set of all countably additive measures. Consequently, $\hat w$ may be discontinuous at the prior. To overcome this challenge, we employ a construction that holds the prior fixed. We refer the reader to the appendix for the exact details.

Our next result uses Theorem~\ref{ubrsuff} to show unique equilibrium S payoffs are a generic feature of finite environments.

\begin{proposition}[Generic uniqueness]\label{generic}
If $A$ and $\Theta$ are finite, then an open dense $\mathcal U_R \subseteq \real^{A\times\Theta}$ of full Lebesgue measure exists such that S has a unique equilibrium payoff (for any prior and any S preferences) as long as $u_R \in \mathcal U_R$.
\end{proposition}

The proposition's proof shows that global PUBR is generic.  Intuitively, a failure of PUBR at some belief implies that, for some fixed action and fixed set of states, R's highest possible expected payoff gain from using said action rather than his best other action is exactly zero---a knife-edge condition.  

The above proposition, which makes no structural assumptions on players' payoffs, is completely silent on infinite persuasion models. Meanwhile, papers applying the Bayesian persuasion framework \citep[e.g.,][]{kolotilin2017persuasion,kolotilin2018optimal,dworczak2018simple,guo2019interval} often focus on settings in which the state space, action space, or both are infinite.\footnote{Moreover, in a setting in which R has private information and takes a binary action, such as \cite{kolotilin2017persuasion}, public persuasion can be reinterpreted as a continuous-action model in which R chooses a cutoff type at which to take the high action.} Because many such models are in some sense one-dimensional, we next turn our attention to specifications enjoying some ordinal structure.

\begin{definition}\label{orderdef} 
Say the environment is \textbf{ordered} if $A,\Theta\subseteq\real$; the function $u_R$ exhibits strictly increasing differences; and every $\mu\in\DT$ has either
$\int u_R(a,\cdot) \ddd\mu$ strictly quasiconcave in $a\in A$ or $\int u_S(a,\cdot) \ddd\mu$ weakly quasiconvex in $a\in A$.\footnote{The quasiconcavity condition says that any $a_L,a_M,a_H\in A$ with $a_L<a_M<a_H$ have $\int u_R(a_M,\cdot) \ddd\mu>\min\left\{\int u_R(a_L,\cdot) \ddd\mu,\ \int u_R(a_H,\cdot) \ddd\mu\right\}$. This condition holds vacuously with binary actions, holds if $u_R$ is strictly concave in its first argument, and holds if $u_R$ is the restriction of a function $A\times\co\Theta\to\real$ that is strictly quasiconcave in its first argument and affine in its second. Similarly, the quasiconvexity condition holds if $u_S$ is weakly convex or weakly monotone in its first argument, or if $u_S$ is the restriction of a function $A\times\co\Theta\to\real$ that is weakly quasiconvex in its first argument and affine in its second.}
\end{definition}

The above condition says that the action space and state space are both subsets of the real line, and that preferences respect this ordered structure. 
The increasing-differences condition says that R would like to take higher actions in higher states, and the quasiconvexity/quasiconcavity condition (which for instance holds if the the first-order approach is valid in determining R's optimal behavior, or if S always prefers higher actions) ensures that S always prefers either R's highest or lowest best response. 
Note, the assumption that the environment is ordered is silent on whether R has unique best responses, on whether preferences are affine in R's posterior belief, and on whether the state or action space is finite or infinite. 

The next result shows that, except under a specific knife-edge condition (on R's preferences and the prior), S has a unique equilibrium payoff. In particular, the result says S has a unique equilibrium payoff in ordered models with an atomless prior.

\begin{theorem}[Uniqueness in the ordered model]\label{ordered}
Suppose the environment is ordered; and for each $\bar\theta\in\{\min\supp\prior,\ \max\supp\prior\}$, either $\mu_0(\bar\theta)=0$ or R's best response at $\delta_{\bar\theta}$ is unique. Then, S has a unique equilibrium payoff.
\end{theorem}

For a proof sketch, it is useful to consider the case in which $A$ is finite. In this special case, one can derive payoff uniqueness as a consequence of Theorem~\ref{ubrsuff}.\footnote{In the general case, Theorem~\ref{ordered} does not follow directly from Theorem~\ref{ubrsuff}. However, the two results are related. In the appendix, we prove a technical sufficient condition (Lemma~\ref{lem: heart of the prior PUBR theorem}) on the functions $v$ and $w$, which serves as a key input to both theorems.} We need to argue that the PUBR property holds at any belief $\mu$ supported by an information policy. 
Given our generic assumption on the prior, we only need to worry about the case of beliefs at which R is neither certain of the highest state nor certain of the lowest state. The quasiconcavity/quasiconvexity condition of ordered models implies that S prefers either S's highest or lowest best response at $\mu$---say the highest. But then the degenerate belief on the highest state is a valid witness to the PUBR property: slightly increasing the probability on the high state breaks any R indifference in S's preferred direction, because of the increasing differences hypothesis. To prove Theorem~\ref{ordered} in the general case, a qualitatively similar argument applies (though one cannot appeal to~Theorem~\ref{ubrsuff}). This argument relies on the model's continuity hypotheses. Instead of perturbing beliefs in the direction of a degenerate belief, we do so in the direction of a \emph{conditional} belief of the prior conditional on the state exceeding a high cutoff; and instead of ensuring an action is selected as a unique best response, we ensure that all R best responses are very near to the targeted action.


To conclude the section, we provide two examples in which our two theorems make it straightforward to verify S's equilibrium payoff is unique (without knowing $\hat v(\mu_0)$ or $\hat w(\mu_0)$).

\begin{example}
Suppose $A\subseteq\Theta\subset\real^d$ for some dimension $d\in\mathbb N$, where $A$ is finite (but $\Theta$ can be finite or infinite), and suppose R's utility takes the form $\ur(a,\theta)=-\Vert a - \theta\Vert^2$. For any full-support prior and any S preferences, the PUBR property holds on $\DT$ because each $a\in A$ is the unique best response to the degenerate belief $\delta_a\in\DT$. Thus, Theorem~\ref{ubrsuff} applies directly, and S has a unique equilibrium payoff.
\end{example}

\begin{example}
Suppose $A\subset\real_{+}$, $\Theta\subset[1,\infty)$,
the prior is atomless, R's utility takes the form $\ur(a,\theta)=\theta^a - c(a)$ for some $c\colon A\to\real$, and $\us$ is increasing in its first argument. This environment is ordered, with strictly increasing differences following from $$\frac{\partial}{\partial a}\frac{\partial}{\partial \theta}\ur(a,\theta)=\theta^{a-1}(1+a\log\theta)>0.$$
Thus, Theorem~\ref{ordered} applies directly, and S has a unique equilibrium payoff.
\end{example}

\section{Limited Commitment}\label{sec:pwi}

One way of arguing in favor of sender-favorite equilibrium selection is to assert that a principal with full commitment power should also have the ability to steer others towards her favorite equilibrium. This argument, however, may appear less compelling in settings with a weaker sender, such as when the sender's ability to commit is compromised. In this section, we show how to apply our results to such settings by extending the model to one in which S has only limited commitment power, as modeled in \cite*{LRS1}---hereafter LRS1. To do so, we first specialize the existing parameters of our environment for tractability: \begin{assumption}
The spaces $A$, $\Theta$, and $M$ are all finite with $|M|\geq2|\Theta|$, and $u_S$ is state independent (i.e., constant in its second argument). 
\end{assumption}
\noindent In mild abuse of notation, write $u_S\colon A\to\real$. In addition to the parameters of our baseline model, we parameterize our limited-commitment model by $\cred\in[0,1]$, which denotes the sender's \emph{credibility}.

The augmented game begins with S provisionally choosing an experiment, $\cstrat\colon \Theta\to\Delta M$, an ``official'' report. The state $\theta\in\Theta$ then realizes and, independent of the state, one of two possibilities occurs. With probability $\cred$, the message $m\in M$ is sent in accordance with the report, and so is distributed according to $\cstrat(\cdot|\theta)$. With complementary probability $1-\cred$, reporting is ``influenced'' and so S observes $\theta$ and can freely choose $m\in M$. Crucially, R observes the message $m$ but observes neither $\theta$ nor whether reporting was influenced.

In what follows, we formalize an appropriate solution concept for the game with compromised reporting (again taking experiment-choice incentives into account), provide a result comparing the S equilibrium payoff set to that under perfect commitment, and apply our uniqueness results to address robustness to imperfect credibility.

\subsection{The equilibrium payoff set}

We formalize the solution concept for the augmented game as follows.

\begin{definition}
A sender strategy is a pair $(\cstrat,\sstrat)$ consisting of an experiment $\cstrat\in\cstrats$ and a measurable function $\sstrat\colon\Theta\times\cstrats\to\Delta M$; 
a receiver strategy is a measurable function $\rstrat\colon M\times\cstratrs\to\Delta A$;
and a receiver belief map is a measurable function $\beliefm\colon M\times\cstratrs\to\DT$.  A \textbf{perfect Bayesian $\cred$-equilibrium ($\cred$-PBE)} is a quadruple of such maps $\langle\cstrat,\sstrat,\rstrat,\beliefm\rangle$ such that \begin{enumerate}
\item The sender's experiment choice satisfies 
$$\cstrat
\in \argmax_{\cstratr\in\cstratrs} \int_\Theta \int_M \int_A u_S \ddd\rstrat(\cdot|m,\cstratr) \ \left[ \cred\ddd\cstratr(m|\theta)+(1-\cred)\ddd\sstrat(m|\theta,\cstratr)\right] \ddd\mu_0(\theta);$$
\item Every $\cstratr\in\cstratrs$ and $m\in M$ have $$\rstrat\left(\argmax_{a\in A} \int_\Theta u_R(a,\theta) \ddd\beliefm(\theta|m,\cstratr) \ \middle|\  m,\cstratr \right)=1;$$
\item Every $\cstratr\in\cstratrs$, Borel $\hat M\subseteq M$, and Borel $\hat\Theta \subseteq \Theta$ have $$\int_\Theta \int_{\hat M} \beliefm(\hat\Theta|\cdot,\cstratr) \ddd\left[ \cred\cstratr(\cdot|\theta)+(1-\cred)\sstrat(\cdot|\theta,\cstratr)\right] \ddd\mu_0(\theta) = \int_{\hat\Theta} \left[ \cred\cstratr(\hat M|\cdot)+(1-\cred)\sstrat(\hat M|\cdot,\cstratr)\right] \ddd\mu_0;$$
\item Every $\cstratr\in\cstratrs$ and $\theta\in\Theta$ have $$\sstrat\left(\argmax_{m\in M} \int_A u_S(a) \ddd\rstrat(a|\cstratr,m) \ \middle|\  \theta,\cstratr \right)=1.$$
\end{enumerate}
In such a case, we say the induced \textbf{$\cred$-PBE payoff} is
$$
\int_\Theta \int_M \int_A u_S(a) \ddd\rstrat\left(a\middle|m,\cstrat
\right) \left[ \cred\ddd\cstrat(m|
\cdot) 
+(1-\cred)\ddd\sstrat(m|\cdot,\cstrat) \right]
\ddd\mu_0
.$$
Let $w^*_\cred(\mu_0)$ denote the infimum $\cred$-PBE payoff for S.
\end{definition}

We begin our analysis of this richer model with a partial description of the range of $\cred$-PBE payoff for S.
With full credibility, the range of S payoffs is naturally identical to the game without an influencing stage (as influence occurs with zero probability): one completes the equilibrium by incorporating some best response for an influencing S. 
But which payoffs can S attain when unable to perfectly commit? With S-favorable selection, and ignoring experiment-choice incentives, it is nearly immediate that higher credibility can only help S: she could always use her additional commitment power to replicate the way should would behave without it.
But the effect of credibility is less clear when one considers S incentives and the full range of equilibria. First, the above replication argument is only available if S finds it optimal to mimic her hypothetical influencing behavior. Second, because influencing S is subject to incentive constraints more often under lower credibility, it is natural to wonder whether her incentive constraints would refine away some bad equilibria.

The following result shows that, in spite of these concerns, lower credibility is in fact worse for S than full commitment power in a set-valued sense.\footnote{Note, the proposition also establishes that a $\cred$-PBE exists.} 

\begin{proposition}[Payoff set under partial credibility]\label{allcred}
The set of $\cred$-PBE payoffs for S is nonempty and weakly below $[\hat w(\mu_0),\hat v(\mu_0)]$ in the strong set order, coinciding with it for $\cred=1$.
\end{proposition}

The part of the proposition for $\cred=1$ is straightforward. We now briefly summarize the proof showing S's $\cred$-PBE payoffs lies below $[\hat w(\mu_0),\hat v(\mu_0)]$ when $\cred<1$. Our explanation takes existence of a $\cred$-PBE as given. As a starting point, note that additional incentive considerations can only constrain S, and so no  $\cred$-PBE payoff can be greater than $\hat v(\mu_0)$. Now, fix any $\cred < 1$, and let $v^*_{\cred}(\mu_0)$ be the highest $\cred$-PBE payoff for S. If $\hat{v}_\cred$ were below $\hat w(\mu_0)$, the desired set ranking immediately follows. Suppose then that $v^*_\cred(\mu_0) > \hat w(\mu_0)$. In this case, we argue that every payoff between $\hat{w}(\mu_0)$ and $v^*_\cred(\mu_0)$ is compatible with $\cred$-PBE. For a rough sketch, consider first S's set of attainable payoffs if we ignored her experiment-choice incentives. Since one can always pair an uninformative official report with a babbling equilibrium, this set is always includes $w(\mu_0)$. Moreover, this set is an interval, by Lemma 7 in LRS1. Since $w(\mu_0)\leq \hat{w}(\mu_0)$, it follows the set of payoffs S can obtain ignoring experiment-choice incentives is an interval whose lower bound is weakly below $\hat{w}(\mu_0)$.

Given the above, to prove the proposition, it is sufficient to claim that for any experiment S could initially choose, some continuation equilibrium gives S a payoff of no more than $\hat w(\mu_0)$.  This claim would follow immediately from the definition of $\hat w$ if R were to choose adversarially to S whenever he is indifferent between multiple actions. However, it is unclear whether a continuation equilibrium with the latter feature exists, since $w$ is not continuous (hence, not upper semicontinuous)---hence S-adversarial tie-breaking by R could mean S does not have any optimal message to send when allowed to influence the state. To resolve this issue, we search for a continuation equilibrium in which R gives S a continuation payoff in $[w(\mu),z(\mu)]$ whenever his posterior belief is $\mu$, where $z$ is the smallest upper semicontinuous function above $w$ (and so $z\leq v$). Because the correspondence $[w,z]$ is upper hemicontinuous, an appropriate application of Kakutani's fixed point theorem delivers such continuation play and beliefs. This continuation play generates an S payoff of at most $\hat z(\mu_0)$ by KG's results, where $\hat z$ is the concave envelope of $z$. Finally, because the concave function $\hat w$ on a finite-dimensional simplex is automatically continuous on the interior, we can show that the concave envelopes $\hat z$ and $\hat w$ agree at the prior.

\subsection{Strong robustness}

LRS1 establishes a robustness result (Proposition 3 of that paper) concerning the best S payoff attainable when credibility is only slightly imperfect ($\cred\approx1$)---but maintaining the assumption of S-favorable equilibrium selection. That result implies her highest attainable payoff converges to the Bayesian persuasion value for most payoff specifications. Our interpretation of this result is that S's value from persuasion is typically robust to limited credibility \emph{if} she can also select the equilibrium. As explained earlier, though, when S's credibility is compromised, one may doubt that she can coordinate R towards her favorite equilibrium. In this section, we revisit the robustness question of LRS1, without assuming S can choose her favorite equilibrium when her credibility is imperfect. In particular, we show that S payoffs are robust to slightly imperfect credibility \emph{and} equilibrium selection if and only if they are robust to equilibrium selection in the full-credibility case. With this result in hand, the results of the previous sections apply directly to address this stronger form of robustness in persuasion models.

\begin{proposition}[Strong robustness]\label{Worst}
The lowest $\cred$-PBE values satisfy $\lim_{\cred\nearrow 1} w^*_\cred(\mu_0)=\hat w(\mu_0)$. In particular, the Bayesian persuasion value is strongly robust to partial credibility ($\lim_{\cred\nearrow 1} w^*_\cred(\mu_0) = \hat v(\mu_0)$) if and only if it is robust to equilibrium selection ($\hat w(\mu_0) = \hat v(\mu_0)$).
\end{proposition}

The proof constructs lower payoff bounds for S in any $\cred$-PBE by computing her payoff following viable deviations, and shows that these payoff bounds can be made arbitrarily close to $\hat w(\mu_0)$ as credibility becomes arbitrarily close to perfect. In brief, we consider the deviation to an experiment that would be approximately optimal under full S commitment when R breaks ties against S's interests. Using lower semicontinuity of $w$, we show S attains a payoff approximating her full-commitment payoff from this experiment as $\cred$ approaches $1$, because the belief a message induces is very nearby the full-commitment version of the same.

Finally, Propositions~\ref{generic}~and~\ref{Worst} immediately yield the following conclusion.

\begin{corollary}[Generic strong robustness]\label{cor: imperfect generic}
For any finite $A$ and $\Theta$, for all but a Lebesgue-null (and nowhere dense) set of R objectives $u_R\in\real^{A\times\Theta}$, and every S objective $u_S\in\real^A$, the full commitment value is strongly robust to partial credibility.
\end{corollary}

Thus, for most payoff specifications, S obtains essentially her Bayesian persuasion payoff in any equilibrium, even if her ability to commit to the information she provides is slightly imperfect.

\end{spacing}

\appendix

\section{Proofs}

Before proceeding to formal proofs, we review for convenience several key notations. \begin{eqnarray*}
A^*_R\colon \DT &\rightrightarrows& A \\
\mu &\mapsto& \argmax_{a\in A} \int u_R(a,\cdot) \ddd\mu \\
V\colon \DT &\rightrightarrows& \real \\
\mu &\mapsto& \co\left\{ \int u_S(a,\cdot) \ddd\mu\colon\, a\in A^*_R(\mu) \right\} \\
v\colon \DT &\to& \real \\
\mu &\mapsto& \max V(\mu) \\
w\colon \DT &\to& \real \\
\mu &\mapsto& \min V(\mu) \\
\RP\colon \DT &\rightrightarrows& \Delta\DT \\
\mu &\mapsto& \left\{ p\in\Delta\DT\colon\, \int \tilde\mu \ddd p(\tilde\mu)=\mu \right\} \\
\hat v\colon \DT &\to& \real \\
\mu &\mapsto& \max_{p\in\RP(\mu)} \int v \ddd p \\
\hat w\colon \DT &\to& \real \\
\mu &\mapsto& \sup_{p\in\RP(\mu)} \int w \ddd p.
\end{eqnarray*}
Note (appealing to Berge's theorem and to continuity of the barycenter map) that all three correspondences are nonempty-compact-valued and upper hemicontinuous, the functions $v$ and $\hat v$ are upper semicontinuous, and the function $w$ is lower semicontinuous.

\subsection{Proofs for Section~\ref{sec:fullset}}
The following lemma shows information policies can be replaced with payoff-equivalent ones of small support when the state space is finite.
\begin{lemma}\label{carath}
If $\Theta$ is finite, then any $\bar\mu\in\DT$, Borel $D\subseteq \DT$, bounded measurable $f\colon\DT\to\real$ and $p\in\RP(\bar\mu)\cap\Delta D$ admit some $q\in\RP(\bar\mu)\cap\Delta D$ with $\supp(q)$ affinely independent (hence of cardinality no greater than $|\Theta|$) and $\int f \ddd q \geq \int f \ddd p$. 
\end{lemma}
\begin{proof}
The set $\RP(\bar\mu)$ is convex compact metrizable, and so Choquet's theorem yields some $Q\in\Delta[\RP(\bar\mu)]$ with barycenter $p$ such that $Q$ is supported on the extreme points of $\RP(\bar\mu)$. But this set $\ext[\RP(\bar\mu)]$ consists exactly of those $q\in \RP(\bar\mu)$ with affinely independent support.

By definition of the barycenter, $\iint g \ddd q \ddd Q(q)=\int g \ddd p$ for every continuous $g\colon\DT\to\real$. However, because the barycenter of this $Q$ is unique, it follows that $p(B)=\int q(B)\ddd Q(q)$ for every Borel $B\subseteq\DT$. Hence, $Q(\Delta D)=1$ because $p(D)=1$, and $\iint g \ddd q \ddd Q(q)=\int g \ddd p$ for every bounded measurable $g\colon\DT\to\real$---in particular for $g=f$.

Letting $\mathcal J\coloneqq\ext[\RP(\bar\mu)]\cap\Delta D$, we have $Q(\mathcal J)=1$. Therefore, $$0 = \iint f \ddd q \ddd Q(q) - \int f \ddd p = \int_{\mathcal J} \left[\int f \ddd q- \int f \ddd p\right] \ddd Q(q).$$ 
So the integrand is somewhere nonnegative: some $q\in\mathcal J$ has $\int f \ddd q\geq \int f \ddd p$.
\end{proof}

Now, we prove the characterization of all S equilibrium payoffs.

\begin{proof}[Proof of Proposition~\ref{perprior}]
To begin, we recall some well-known facts about experiments and Bayesian updating---which collectively tell us that S choosing from $\cstratrs$ and choosing from $\RP(\mu_0)$ are equivalent formalisms.  First, any experiment $\cstratr\in\cstratrs$ admits some compatible belief map, that is, some measurable $\tilde\beliefm=\tilde\beliefm_\cstratr\colon M\to\DT$ such that every Borel $\hat M\subseteq M$ and Borel $\hat\Theta \subseteq \Theta$ have $\int_\Theta \int_{\hat M} \tilde\beliefm(\hat\Theta|m,\cstratr) \ddd\cstratr(m|\theta) \ddd\mu_0(\theta) = \int_{\hat\Theta} \cstratr(\hat M|\theta) \ddd\mu_0(\theta).$ Second, given $\cstratr\in\cstratrs$ if we define the belief distribution $p_{\cstratr,\tilde\beliefm} \in \Delta\DT$ via $p_{\cstratr,\tilde\beliefm}(D)\coloneqq\int_\Theta  \cstratr\left( \tilde\beliefm^{-1}(D) \ \middle|\ \theta\right) \ddd\mu_0(\theta)$ for each Borel $D\subseteq\DT$, then $p_{\cstratr,\tilde\beliefm}=p_{\cstratr,\tilde\beliefm'}$ for any two such compatible $\tilde\beliefm$ and $\tilde\beliefm'$.  We therefore refer to the associated belief distribution simply as $p_\cstratr$.  Third, every $\cstratr\in\cstratrs$ has $p_\cstratr\in\RP(\mu_0)$.  Fourth, every $p\in\RP(\mu_0)$ with $|\supp(p)|\leq|M|$ admits some $\cstratr_p\in\cstratrs$ such that $p_{\cstratr_p}=p$. 

Now we proceed to show $s\in[\hat w(\mu_0), \hat v(\mu_0)]$ is necessary and sufficient for $s$ to be an equilibrium S payoff.

First, to see the condition is necessary, fix an arbitrary equilibrium $\langle\cstrat,\rstrat,\beliefm\rangle$, and let $s\in\real$ be the induced S payoff; we will show $s\in [\hat w(\mu_0), \hat v(\mu_0)]$.  For any $\cstratr\in\cstratrs$ and $m\in M$ the R optimality condition implies $\rstrat\left(A^*_R\left(\beliefm(\cdot|m,\cstratr)\right) \ \middle|\  m,\cstratr \right)=1$, so that $\int_A u_S(a,\theta) \ddd\rstrat\left(a\middle|m,\cstratr\right) \in V\left(\beliefm(\cdot|m,\cstratr)\right)$. Therefore, any $\cstratr\in\cstratrs$ has 
\begin{eqnarray*}
\int_\Theta \int_M \int_A u_S(a,\theta) \ddd\rstrat(a|m,\cstratr) \ddd\cstratr(m|\theta)\ddd\mu_0(\theta) &\in&  \int_\Theta \int_M V\left(\beliefm(\cdot|m,\cstratr)\right)
\ddd\cstratr(m|\theta)\ddd\mu_0(\theta) \\
&=& \int_\DT V \ddd p_\cstratr \\
&=& \left[ \int_\DT w \ddd p_\cstratr, \ \int_\DT v \ddd p_\cstratr\right] \\
&\subseteq& \left[ \int_\DT w \ddd p_\cstratr, \ \hat v(\mu_0)\right].
\end{eqnarray*}
Hence, $s\leq \hat v(\mu_0)$.  Moreover, for any $p\in\RP(\mu_0)$ if $\Theta$ is infinite, and for any $p\in\RP(\mu_0)$ such that $|\supp(p)|\leq|\Theta|$ if $\Theta$ is finite, taking $\tilde\cstrat=\cstratr_p$ implies (by S rationality) \begin{eqnarray*}
&& \int_\Theta \int_M \int_A u_S(a,\theta) \ddd\rstrat(a|m,\cstrat) \ddd\cstrat(m|\theta)\ddd\mu_0(\theta) \\
&\geq& \int_\Theta \int_M \int_A u_S(a,\theta) \ddd\rstrat(a|m,\tilde\cstrat) \ddd\cstratr(m|\theta)\ddd\mu_0(\theta) \\
&\geq& \int_\DT w \ddd p.
\end{eqnarray*}
Applying this observation to every such $p$ (and applying Lemma~\ref{carath} if $\Theta$ is finite) 
implies $s\geq\hat w(\mu_0)$.

Conversely, take any $s\in [\hat w(\mu_0), \hat v(\mu_0)]$.  
Letting $p_1\in\RP(\mu_0)$ with $\int_{\DT} v \ddd p_v=\hat v(\mu_0)$ and $|\supp(p_v)|\leq|\Theta|$ if $\Theta$ is finite---which exists by Lemma~\ref{carath}---define $p_\lambda\coloneqq\int \delta_{\lambda\mu+(1-\lambda)\mu_0}\ddd p_1(\mu)\in \RP(\mu_0)$ for each $\lambda\in[0,1]$; observe $|\supp(p_\lambda)|\leq|\supp(p_1)|\leq|M|$. As $\lambda\mapsto p_\lambda$ is continuous, it follows that the $\lambda\mapsto \int V \ddd p_\lambda$ is nonempty-compact-convex-valued and upper hemicontinuous because $V$ is. Moreover, $$\int V \ddd p_0= V(\mu_0)\ni w(\mu_0) \leq s \leq \hat v(\mu_0)\in \int V \ddd p_1.$$ The intermediate value theorem for correspondences \citep[e.g., Lemma 2 from][]{de2008axiomatization} therefore delivers some $\lambda\in[0,1]$ such that $s\in \int V \ddd p_\lambda$. Some measurable $\zeta\colon\DT\to[0,1]$ then exists such that $s = \int \left[(1-\zeta) w + \zeta v\right] \ddd p_\lambda$. 
By the measurable maximum theorem \cite[Theorem 18.19 from][]{aliprantis2006infinite}, a pair of measurable functions $\alpha_w,\alpha_v\colon\DT\to\Delta A$ exist such that, each $\mu\in\DT$ has $\int_{A\times\Theta} u_S \ddd \left[\alpha_w(\cdot|\mu) \otimes\mu\right] = w(\mu)$, $\int_{A\times\Theta} u_S \ddd \left[\alpha_v(\cdot|\mu) \otimes\mu\right] = v(\mu)$, and $\alpha_w\left( A^*_R(\mu) \middle|\mu\right)=\alpha_v\left( A^*_R(\mu) \middle|\mu\right)=1$.  
With these objects in hand, we can define our candidate $\cstrat:
\Theta\to\Delta M$, $\rstrat\colon M\times\cstratrs\to\Delta A$, and $\beliefm\colon M\times\cstratrs\to\DT$ via \begin{eqnarray*}
\cstrat(\cdot|
\theta) &\coloneqq & \cstratr_{p_\lambda} \\
\beliefm(\cdot|m,\cstratr)&\coloneqq & \tilde\beliefm_\cstratr(\cdot|m) \\
\rstrat(\cdot|m,\cstratr) &\coloneqq &\begin{cases}
(1-\zeta)\alpha_w \left(\cdot\ \middle| \ \beliefm(\cdot|m,\cstratr) \right) + \zeta\alpha_v \left(\cdot\ \middle| \ \beliefm(\cdot|m,\cstratr) \right) & : \ \cstratr=\cstratr_{p_\lambda} \\
\alpha_w\left(\cdot\ \middle| \ \beliefm(\cdot|m,\cstratr) \right)& : \ \cstratr\neq\cstratr_{p_\lambda}.
\end{cases}
\end{eqnarray*}
It is immediate from the construction that all three maps are measurable and that R rationality and the Bayesian property are both satisfied.  Moreover, direct computation shows that choosing experiment $\cstratr\in\cstratrs$ gives S a continuation payoff of $$\int_\Theta \int_M \int_A u_S(a,\theta) \ddd\rstrat(a|m,\cstratr) \ddd\cstratr(m|\theta)\ddd\mu_0(\theta)=\begin{cases}
s &: \ \cstratr=\cstratr_{p_\lambda} \\
\int_{\DT} w \ddd p_\cstratr &: \ \cstratr\neq\cstratr_{p_\lambda} \\
\end{cases}
$$
Therefore, S gets payoff $s$ if the triple is an equilibrium.  Finally, the triple is indeed an equilibrium: In particular, S rationality is confirmed because $s\geq \hat w(\mu_0) \geq \int w \ddd p_\cstratr$ for every alternative $\cstratr\in\cstratrs$.
\end{proof}

\subsection{Proofs for Section~\ref{sec:unique}}

\begin{lemma}\label{lem: hat w is idempotent}
The function $\hat w\colon\DT\to\real$ is lower semicontinuous (hence measurable) and concave, and every $\bar p\in\RP(\prior)$ has $\hat w(\prior)\geq\int \hat w\ddd \bar p$.
\end{lemma}
\begin{proof}
First, the function $\Delta\DT\to\real$ given by $p\mapsto \int w\ddd p$ is lower semicontinuous because $w$ is. Moreover, the barycenter map $\Delta\DT\to\DT$ is open \cite[][Corollary 1]{o1976openness}, meaning $\RP$ is lower hemicontinuous. It follows from \cite[][Lemma 17.29]{aliprantis2006infinite} that $\hat w$ is lower semicontinuous. In the next paragraph we establish the last assertion in the lemma's statement. Specializing that assertion to the case where $\bar p$ has finite support (and varying $\mu_0$) delivers that $\hat w$ is concave. The lemma follows.

Fix an arbitrary information policy $\bar p\in\RP(\prior)$ and $\varepsilon>0$; we want to show $\hat w(\prior)\geq \int \hat w \ddd \bar p - \varepsilon$. To that end, define the correspondence 
\begin{eqnarray*}
\Phi\colon \DT &\rightrightarrows& \Delta\DT \\
\mu &\mapsto& \left\{p\in\RP(\mu)\colon\, \int w \ddd p \geq \hat w(\mu)-\varepsilon\right\}.
\end{eqnarray*}
This correspondence is nonempty-valued by definition of $\hat w$, and has measurable graph because $\hat w$ is measurable and the barycenter map $p\mapsto\int\mu\ddd p(\mu)$ is continuous (hence measurable). Therefore, Corollary 18.27 from \cite{aliprantis2006infinite} yields a measurable function $\varphi\colon\DT\to\Delta\DT$ with $\bar p\{\varphi\in\Phi\}=1$. Because $\bar p\in\RP(\prior)$ and $\bar p\{\varphi\in\RP\}=1$, we know $p\coloneqq\int\varphi\ddd\bar p$ is in $\RPP$. Therefore, $$\hat w(\prior)\geq \int w\ddd p = \iint w \ddd\varphi(\cdot|\mu)\ddd\bar p(\mu) \geq\int (\hat w -\varepsilon)\ddd\bar p=\int \hat w \ddd\bar p-\varepsilon$$
The lemma follows. 
\end{proof}

\begin{proof}[Proof of Proposition~\ref{infoselect}]
Let $\RP_{D}(\prior)$ be the set of all information policies  $p \in \RP(\prior)$ such that $p(\tilde D) = 1$ for some measurable $\tilde D \subseteq D$. Because $D$ is persuasion sufficient and $\hat w|_D\geq v|_D$, and by Lemma~\ref{lem: hat w is idempotent}, we have:
$$\hat v(\prior)=\sup_{p\in\RP_{D}(\prior)} \int v\ddd p\leq\sup_{p\in\RP_{D}(\prior)} \int \hat w\ddd p\leq \hat w(\prior).$$
The proposition follows.
\end{proof}

\begin{proof}[Proof of Corollary~\ref{allprior}]
By Proposition~\ref{perprior}, it suffices to show that $\hat w = \hat v$ if and only if $\hat w\geq v$.  As it is immediate that $\hat v \geq \hat w$ and $\hat v \geq v$, we need only see that $\hat w \geq \hat v$ if $\hat w\geq v$. 
But this result follows directly from Proposition~\ref{infoselect}, because $\DT$ is vacuously persuasion sufficient.
\end{proof}

\begin{lemma}\label{lem: heart of the local PUBR theorem}
Suppose $D\subseteq\DT$ is measurable, $p\in\RPP$ has $p(D)=1$, and $f\colon\DT\to\real$ is bounded and measurable. Suppose every $\bar\mu\in D$ admits some $\mu\in\DT$ and $\gamma>0$ such that $\gamma\mu\leq\bar\mu$ and all sufficiently small $\lambda\in(0,1]$ have
$\hat w\left(\lambda\mu+(1-\lambda)\bar\mu\right)\geq f(\bar\mu)$. Then $\hat w(\mu_0)\geq \int f\ddd p$.
\end{lemma}
\begin{proof}
If we establish $\hat w|_D\geq f|_D$, the lemma will follow from Lemma~\ref{lem: hat w is idempotent}, because we will have $\hat w(\prior)\geq\int \hat w\ddd p \geq \int f \ddd p$. So take any $\bar\mu\in D$, with a view to showing $\hat w(\bar\mu)\geq f(\bar\mu)$.

Let $\mu\in\DT$ and $\gamma>0$ be as guaranteed by the lemma's statement.
Assuming without loss that $\gamma\in(0,1)$, let $\mu'\coloneqq\tfrac{\bar\mu-\gamma\mu}{1-\gamma}\in\DT$ so that $\bar\mu=\gamma\mu+(1-\gamma)\mu'$. Then any sufficiently small $\gamma\in(0,1]$ has \begin{eqnarray*}
    \bar\mu &=& \tfrac{\gamma}{\gamma+\lambda(1-\gamma)}\left[\lambda\mu+(1-\lambda)\bar\mu\right]+\tfrac{\lambda(1-\gamma)}{\gamma+\lambda(1-\gamma)} \mu' \\
    \implies\hat w(\bar\mu) &\geq& \tfrac{\gamma}{\gamma+\lambda(1-\gamma)}\hat w\left(\lambda\mu+(1-\lambda)\bar\mu\right)+\tfrac{\lambda(1-\gamma)}{\gamma+\lambda(1-\gamma)} \hat w(\mu') \text{ (by Lemma~\ref{lem: hat w is idempotent})}\\
    &\geq& \tfrac{\gamma}{\gamma+\lambda(1-\gamma)}f(\bar\mu)+\tfrac{\lambda(1-\gamma)}{\gamma+\lambda(1-\gamma)} \min u_S(A\times\Theta) \\
    &\to& f(\bar\mu) \text{ as } \lambda\to0.
 \end{eqnarray*}
Thus, $\hat w(\bar\mu)\geq f(\bar\mu)$, delivering the lemma.
\end{proof}

\begin{lemma}\label{lem: heart of the prior PUBR theorem}
Suppose $D\subseteq\DT$ is measurable, $p\in\RPP$ has $p(D)=1$, and $f\colon\DT\to\real$ is bounded and measurable. Suppose every $\bar\mu\in D$ admits some $\mu\in\DT$ and $\gamma\in(0,1)$ such that $\gamma\mu\leq\prior$ and all sufficiently small $\lambda\in(0,1]$ have $\hat w\left(\lambda\mu+(1-\lambda)\bar\mu\right)>f(\bar\mu)$. Then, $\hat w(\mu_0)\geq \int f\ddd p$.
\end{lemma}
\begin{proof}
Define the correspondences
\begin{eqnarray*}
F\colon D&\rightrightarrows&\DT \\
\bar\mu &\mapsto& \left\{\mu\in\DT\colon\, \hat w(\mu)>f(\bar\mu)\right\}\\
G\colon D&\rightrightarrows&\DT \\
\bar\mu &\mapsto& \left\{ \mu\in\DT\colon\, 
\lambda\mu+(1-\lambda)\bar\mu\in F(\bar\mu) \ \forall \text{ sufficiently small } \lambda\in(0,1]\right\}\\
H\colon D&\rightrightarrows&\DT\times(0,1) \\
\bar\mu &\mapsto& \left\{(\mu,\zeta)\in G(\bar\mu)\times(0,1)\colon\, \zeta\mu\leq\prior\right\}.
\end{eqnarray*}
By hypothesis, $H$ is nonempty-valued.  At the end of this proof, we shall show that $H$ has measurable graph. Before doing so, let us establish the lemma taking this fact for granted. 

Extend $H$ to $H\colon\DT\rightrightarrows\DT \times (0,1)$ by letting $H(\bar\mu)\coloneqq\DT\times(0,1)$ for every $\bar\mu\in(\DT)\setminus D$. Because $H$ has a measurable graph, the graph of the extended correspondence is also measurable, and so Corollary~18.27 from \cite{aliprantis2006infinite} delivers measurable functions $\varphi\colon\DT\to\DT$ and $\gamma\colon\DT\to(0,1)$ with $p\{(\varphi,\gamma)\in H\}=1$. Let $\xi\coloneqq\int(1-\gamma)\ddd p\in(0,1)$.

For any $\lambda\in(0,1]$, define the measurable function \begin{eqnarray*}
\varphi_\lambda\colon D &\to& \DT \\
\mu &\mapsto& \dfrac1{(1-\lambda)+\lambda\gamma(\mu)} \left[ (1-\lambda)\mu + \lambda\gamma(\mu)\varphi(\mu) \right],
\end{eqnarray*}
and define the measures\footnote{It is immediate from the definition of $\xi$ that $\mu_\lambda\in\ca(\Theta)$ has $\mu_\lambda(\Theta)=1$. To see the measure is positive, hence in $\DT$, observe that $p\in\RPP$ and $p\{(\varphi,\gamma)\in H\}=1$ implies
$$\lambda\xi \mu_\lambda 
= \prior - \int \left[ (1-\lambda)\mu + \lambda\gamma(\mu)\varphi(\mu) \right] \ddd p(\mu)
= \prior - \int \left[ (1-\lambda)\prior + \lambda\gamma\varphi \right] \ddd p
= \lambda\int \left( \prior - \gamma\varphi \right) \ddd p \geq0.
$$
}
$$\mu_\lambda\coloneqq \dfrac{1}{\lambda\xi}\left\{\prior - \int \left(1-\lambda+\lambda\gamma\right)\varphi_\lambda \ddd p \right\} \in\DT$$
and\footnote{Recall $\delta_\mu$ is the degenerate probability measure on $\{\mu\}$ for any $\mu$. Therefore, for any measurable $\tilde D\subseteq\DT$, we have
$p_\lambda(\tilde D)= \int \left[1-\lambda+\lambda\gamma(\mu)\right]\mathbf1_{\varphi_\lambda(\mu)\in \tilde D} \ddd p(\mu) + \lambda\xi\mathbf1_{\mu_\lambda \in \tilde D}.$
}
$$
p_\lambda\coloneqq \int \left(1-\lambda+\lambda\gamma\right)\delta_{\varphi_\lambda(\cdot)} \ddd p + \lambda\xi \delta_{\mu_\lambda} \in \Delta\DT.
$$
The definition of $\mu_\lambda$ ensures $p_\lambda\in\RPP$ for $\lambda\in(0,1]$. Therefore,\footnote{Note Fatou's lemma applies because $\left(1-\lambda+\lambda\gamma\right) w\circ\varphi_\lambda$ is bounded and $p$ is finite.}
\begin{eqnarray*}
\hat w(\prior)&\geq& \liminf_{\lambda\to0} \int \hat w \ddd p_\lambda \ \ \ \ \ \ \ \ \ \ \ \ \ \ \ \ \ \ \ \ \ \ \ \ \ \ \ \ \ \ \ \ \ \ \ \ \ \ \ \ \ \ \ \ \  \text{ (by Lemma~\ref{lem: hat w is idempotent})}\\
&\geq& \liminf_{\lambda\to0} \left[  \int \left(1-\lambda+\lambda\gamma\right) \hat w\circ\varphi_\lambda \ddd p + \lambda\xi \min u_S(A\times\Theta) \right] \\
&=& \liminf_{\lambda\to0} \int \left(1-\lambda+\lambda\gamma\right) \hat w\circ\varphi_\lambda \ddd p  \\
&\geq& \int \liminf_{\lambda\to0} \left[ \left(1-\lambda+\lambda\gamma\right) \hat w\circ\varphi_\lambda \right] \ddd p  \ \ \ \ \ \ \  \text{ (by Fatou's lemma)} \\
&\geq& \int f \ddd p \ \ \ \ \ \ \ \ \ \ \ \ \ \ \ \ \ \ \ \ \ \ \ \ \ \ \ \ \ \ \ \ \ \ \ \ \ \ \ \ \ \ \ \  \text{ (because $p\{\varphi\in G\}=1$)}.
\end{eqnarray*}
The lemma then follows if we show $H$ has measurable graph.

Toward measurability of the graph $\gr H$, let $X\coloneqq D\times\DT$. First note that $\gr F=\{(\bar\mu,\mu)\in X\colon\, f(\bar\mu)-\hat w(\mu)<0\}$ is measurable because $f$ is measurable and $\hat w$ is (by Lemma~\ref{lem: hat w is idempotent}) lower semicontinuous.
Next, consider the graph of $G$. For any $\bar\mu\in D$, that $F$ is (by Lemma~\ref{lem: hat w is idempotent}) convex-valued and $(0,1]=\co\left\{\tfrac1n\colon\, n\in\mathbb N\right\}$ implies $$G(\bar\mu)=\left\{ \mu\in\DT\colon\, 
\tfrac1n\mu+(1-\tfrac1n)\bar\mu\in F(\bar\mu) \ \forall \text{ sufficiently large } n\in\mathbb N\right\}.$$ 
Moreover, for each $n\in\mathbb N$, the map $\psi_n\colon X\to X$ given by $\psi_n( \bar\mu, \mu)\coloneqq \left( \bar\mu, \tfrac1n\mu+(1-\tfrac1n)\bar\mu \right)$ is continuous, hence measurable. Therefore,
\begin{eqnarray*}
\gr G&=& \left\{(\bar\mu,\mu)\in X\colon\, \left( \bar\mu, \tfrac1n\mu+(1-\tfrac1n)\bar\mu \right)\in\gr F \ \forall \text{ sufficiently large } n\in\mathbb N \right\} \\
&=&
\bigcup_{N=1}^\infty\bigcap_{n=N}^\infty
\psi_n^{-1}(\gr F),
\end{eqnarray*}
which is measurable. 
Finally, the graph
$$\gr H = \left[(\gr G)\times(0,1)\right]\cap\{(\bar\mu,\mu,\zeta)\in X\times (0,1)\colon\, \prior-\zeta\mu\geq0\}$$
is measurable too because $(\bar\mu,\mu,\zeta)\mapsto \prior-\zeta\mu$ is continuous and $\ca_+(\Theta)$ is closed in $\ca(\Theta)$.\footnote{
Suppose $(\eta_\beta)_\beta$ is a net in $\ca_+{(\Theta)}$ that weak* converges to some $\eta\in\ca(\Theta)$. For each $\beta$, some $\lambda_\beta\in\real_+$ and $\mu_\beta\in\DT$ have $\eta_\beta=\lambda_\beta\mu_\beta$. We want to show $\eta\geq0$. Passing to a subnet, we may assume (because $\DT$ is compact) $\mu_\beta$ converges to some $\mu\in\DT$; and because the constant function $\mathbf1$ is continuous, $\eta_\beta\to\eta$ implies $\lambda_\beta\to\lambda\coloneqq \int\mathbf1\ddd\eta$. Therefore, $\lambda\geq0$ and so $\eta=\lambda\mu\geq0$.
} The lemma follows.
\end{proof}

\begin{lemma} \label{lem: ubr preserved on convex combinations}
Let $\bar\mu,\mu\in\DT$ and $\varepsilon>0$. If $A^*_R(\mu)=\{a\}$ for some $a\in A^*_R(\bar\mu)$ with $\int u_S(a,\cdot)\ddd\bar\mu > v(\bar\mu)- \varepsilon$, then all sufficiently small $\lambda\in(0,1]$ have $w\left(\lambda\mu+(1-\lambda)\bar\mu\right)> v(\bar\mu)-\varepsilon$.
\end{lemma}
\begin{proof}
Because $a$ is a best response to $\bar\mu$ and a unique best response to $\mu$, it is (by linearity of expected utility) a unique best response to any proper convex combination $\mu'$, yielding $w(\mu')=\int u_S(a,\cdot)\ddd\mu'$, which converges to $v(\mu)$ as $\mu'\to\bar\mu$.
\end{proof}

\begin{lemma}\label{lem: the prior prior PUBR theorem}
S has a unique equilibrium payoff if some persuasion-sufficient $D\subseteq\DT$ exists such that every $\bar\mu\in D$ has $$\max_{a\in A^*_R(\bar\mu)}\int u_S(a,\cdot)\ddd\bar\mu =\sup_{a\in A^*_R(\bar\mu)\cap A^U(\prior)}\int u_S(a,\cdot)\ddd\bar\mu.$$
\end{lemma}
\begin{proof}
Fixing $\varepsilon>0$ and $\bar\mu\in D$, Lemma~\ref{lem: heart of the prior PUBR theorem} means it suffices to find $\mu\in\DT$ and $\gamma>0$ such that $\gamma\mu\leq\prior$ and all sufficiently small $\lambda\in(0,1]$ have $\hat w\left(\lambda\mu+(1-\lambda)\bar\mu\right)> v(\bar\mu)-\varepsilon$. 
Some $a\in A^*_R(\bar\mu)$ has $\int u_S(a,\cdot)\ddd\bar\mu>v(\bar\mu)-\varepsilon$. Then, by hypothesis, some $\mu\in\DT$ and $\gamma>0$ exist such that $\gamma\mu\leq\prior$ and $A^*_R(\mu)=\{a\}$. This $\mu$ is as desired by Lemma~\ref{lem: ubr preserved on convex combinations}.
\end{proof}

\begin{proof}[Proof of Theorem~\ref{ubrsuff}]
First, we establish that the strong PUBR property implies S has a unique equilibrium payoff. Given $\varepsilon>0$ and $\bar\mu\in D$, Lemma~\ref{lem: heart of the local PUBR theorem} means it suffices to find $\mu\in\DT$ and $\gamma>0$ such that $\gamma\mu\leq\bar\mu$ and all sufficiently small $\lambda\in(0,1]$ have $\hat w\left(\lambda\mu+(1-\lambda)\bar\mu\right)> v(\bar\mu)-\varepsilon$. 
To do so, note some $a\in A^*_R(\bar\mu)$ has $\int u_S(a,\cdot)\ddd\bar\mu>v(\bar\mu)-\varepsilon$. Then, by hypothesis, some $\mu\in\DT$ and $\gamma>0$ exist such that $\gamma\mu\leq\bar\mu$ and $A^*_R(\mu)=\{a\}$. This $\mu$ is as desired by Lemma~\ref{lem: ubr preserved on convex combinations}.

Now, we show the PUBR property implies S has a unique equilibrium payoff if $\min\{|A|,|\Theta|\}<\infty$. For both cases, we appeal to Lemma~\ref{lem: the prior prior PUBR theorem}. In light of that lemma, it suffices to note that $A^U\subseteq A^U(\prior)$ if either $\Theta$ or $A$ is finite. If $\Theta$ is finite, the result follows because every belief $\mu\in\Delta[\supp\prior]$ has $\gamma\mu\leq\prior$ for $\gamma\coloneqq\min_{\theta\in\supp\mu}\tfrac{\prior(\theta)}{\mu(\theta)}$.\footnote{A more direct proof (as described in the main text) is available for the case of finite states. Because we require Lemma~\ref{lem: the prior prior PUBR theorem} to prove the finite-action case and Theorem~\ref{ordered}, we find it convenient to use it here too.} If $A$ is finite, the result follows because the set of beliefs in $\Delta[\supp\prior]$ to which a given action is a unique best response is relatively open, and \citep[Lemma 2,][]{lipnowski2018disclosure} the set of beliefs $\mu$ admitting $\gamma>0$ with $\gamma\mu\leq\prior$ is dense in $\Delta[\supp\prior]$.
\end{proof}

\begin{lemma}\label{ubrgeneric}
Given finite $A$ and $\Theta$, define the set $\mathcal U_R \subseteq \real^{A\times\Theta}$ of R objectives $u_R$ such that every $\mu\in\DT$ and $a\in A^*_R(\mu)$ have some $\mu'\in\Delta[\supp\mu]$ with $A^*_R(\mu')=\{a\}$. Then $\mathcal U_R$ is open and dense with full Lebesgue measure.
\end{lemma}
\begin{proof}
Toward showing these properties, let us note an algebraic characterization of $\mathcal U_R$.  
Define the finite index set $\mathbb I\coloneqq\{(a,\hat\Theta)\colon\, a\in A,\ \emptyset\neq \hat\Theta\subseteq\Theta\}$ and, for each $i=(a,\hat\Theta)\in\mathbb I$, define \begin{eqnarray*}
\varphi_i\colon \real^{A\times\Theta} &\to& \real \\
u_R &\mapsto& \max_{\mu\in\Delta\hat\Theta} \min_{a'\in A\setminus\{a\}} \int \left[ u_R(a,\cdot)-u_R(a',\cdot) \right] \ddd\mu.
\end{eqnarray*}
Then $\mathcal U_R=\{u_R \in \real^{A\times\Theta}\colon\, \varphi_i(u_R) \text{ is nonzero for every } i\in\mathbb I\} = \bigcap_{i\in\mathbb I} \varphi_i^{-1}\left(\real\setminus\{0\}\right).$

We can therefore show $\mathcal U_R \subseteq \real^{A\times\Theta}$ is open, dense, and of full Lebesgue measure by establishing that $\varphi_i^{-1}\left(\real\setminus\{0\}\right)$ enjoys these properties for every $i=(a,\hat\Theta) \in \mathbb I$. First, note it is open because (by Berge's theorem) $\varphi_i$ is continuous.  To show it is of full measure (hence also dense), define $\bar z \coloneqq [\mathbf1_{\tilde a = a}]_{\tilde a \in A, \ \tilde\theta\in\Theta} \in \real^{A\times\Theta}$, and observe that $\varphi_i(u_R+\lambda \bar z) = \varphi_i(u_R) + \lambda$ for any $u_R\in\real^{A\times\Theta}$ and any $\lambda\in\real$.  Now, fixing some $\bar\theta\in\Theta$, observe that we can decompose the vector space of all R objectives as the direct sum $\real^{A\times\Theta}=Y\oplus Z$, where $$Y\coloneqq\{u_R\in\real^{A\times\Theta}\colon\, u_R(a,\bar\theta)=0\} \text{ and } Z \coloneqq \{\lambda\bar z\colon\, \lambda\in\real\}.$$
As $\varphi_i(y+\lambda\bar z)=\varphi(y)+\lambda$ for any $y\in Y$ and $\lambda\in\real$, and a singleton is Lebesgue-null in $\real\cong Z$, it follows from the law of iterated expectations that $\varphi_i^{-1}\left(0\right)$ is Lebesgue-null as well.  The proposition follows.
\end{proof}


\begin{proof}[Proof of Proposition~\ref{generic}]
Lemma~\ref{ubrgeneric} names a set $\mathcal U_R$ and shows it is open and dense in $\real^{A\times\Theta}$ with full Lebesgue measure. Meanwhile, by definition of $\mathcal U_R$, the PUBR property holds on $\DT$ whenever $u_R\in\mathcal U_R$. The proposition then follows from Theorem~\ref{ubrsuff}.
\end{proof}


\begin{lemma}\label{lem: heart of the ordered theorem}
Suppose $A,\Theta\subseteq\real$; the function $u_R$ exhibits strictly increasing differences; every $\mu\in\DT$ admits some $a\in\{\min A^*_R(\mu),\ \max A^*_R(\mu)\}$  such that $\int u_S(a,\cdot)\ddd\mu=v(\mu)$; and for each $\bar\theta\in\{\min\supp\prior,\ \max\supp\prior\}$, either $\mu_0(\bar\theta)=0$ or R's best response at $\delta_{\bar\theta}$ is unique. Then, S has a unique equilibrium payoff.
\end{lemma}
\begin{proof}
Let $\theta_L\coloneqq\min\supp\prior$ and $\theta_H\coloneqq\max\supp\prior$, and define the measurable set of beliefs $$D\coloneqq\left(\Delta[\supp\prior]\right)\setminus\left\{\delta_\theta\colon\, \theta\in\{\theta_L,\theta_H\},\  \prior(\theta)=0 \right\}.$$
Clearly, that $p\in\RPP$ implies $p(D)=1$, and so $D$ is persuasion sufficient. 

Fix any $\bar\mu\in D$ and $\varepsilon>0$. Lemma~\ref{lem: heart of the prior PUBR theorem} will deliver the present lemma if we can find some $\mu\in\DT$ and $\gamma\in(0,1)$ such that $\gamma\mu\leq\prior$ and $\hat w(\lambda\mu+(1-\lambda)\bar\mu)>v(\bar\mu)-\varepsilon$ for all sufficiently small $\lambda\in(0,1]$. If $\bar\mu=\delta_{\bar\theta}$ for some $\bar\theta\in\{\theta_L, \theta_H\}$, then (since $\bar\mu\in D$) $\mu_0(\bar\theta)>0$ and $|A^*_R(\bar\mu)|=1$, so that $(\mu,\gamma)=(\bar\mu,\prior(\bar\theta))$ is as desired. We therefore focus on the complementary case in which $\bar\mu$ lies strictly between $\delta_{\theta_L}$ and $\delta_{\theta_H}$ in the sense of first-order stochastic dominance.

By hypothesis, some $\bar a\in \{\min A^*_R(\bar\mu),\ \max A^*_R(\bar\mu)\}$ has $\int u_S(\bar a,\cdot) \ddd\mu=v(\bar\mu)$. Without loss (because our hypotheses are invariant to reversing the order on both $A$ and $\Theta$), say $\bar a=\max A^*_R(\mu)$. 
That $\bar\mu\in D$ with $|A^*_R(\bar\mu)|>1$ implies, given the theorem's hypothesis on the prior, that $\bar\mu\neq\delta_{\theta_H}$. 

Now, for any $\theta^*\in(\theta_L,\theta_H)$, let $\gamma_{\theta^*}\coloneqq\prior(\theta^*,\theta_H]\in(0,1)$ and define the conditional distribution $\mu_{\theta^*}\coloneqq\tfrac1{\gamma_{\theta^*}}\prior\left[(\cdot)\cap(\theta^*,\theta_H]\right]\in\DT$. By construction, we know $\gamma_{\theta^*}\mu_{\theta^*}\leq\prior$ for every $\theta^*\in(\theta_L,\theta_H)$. We will show that $(\mu,\gamma)=(\mu_{\theta^*},\gamma_{\theta^*})$ are as desired when $\theta^*\in(\theta_L,\theta_H)$ is large enough.

Let $\eta>0$ and neighborhood $B\subseteq\DT$ of $\bar\mu$ be small enough that any $a\in A$ with $|a-\bar a|<\eta$ and any $\tilde\mu\in B$ have $\int u_S(a,\cdot) \ddd\tilde\mu>v(\bar\mu)-\varepsilon$. We now claim that, for all sufficiently large $\theta^*\in(\theta_L,\theta_H)$, every $a\in A$ with $a\leq \bar a-\eta$ has $\int u_S(a,\cdot)\ddd\mu_{\theta^*} < \int u_S(\bar a,\cdot)\ddd\mu_{\theta^*}$. 
Assume otherwise, for a contradiction. Then ($A$ being compact) some sequence $(\theta^*_n,a_n)_n$ from $(\theta_L,\theta_H)\times A$ converging to $(\theta_H,a')$, for $a'\in A\cap(-\infty,\bar a-\eta]$, has $a_n\in A^*_R(\mu_{\theta^*_n})$ for every $n$. But $\int [u_S(a_n,\cdot)-u_S(\bar a,\cdot)]\ddd\mu_{\theta^*_n}\geq 0$ converges to $\int [u_S(a',\cdot)-u_S(\bar a,\cdot)]\ddd\delta_{\theta_H}$, which is strictly negative due to strictly increasing differences and $\bar a\in A^*(\bar\mu)$---a contradiction. 

In what follows, fix $\theta^*\in(\theta_L,\theta_H)$ such that every $a\in A$ with $a\leq \bar a-\eta$ has $\int u_S(a,\cdot)\ddd\mu_{\theta^*} < \int u_S(\bar a,\cdot)\ddd\mu_{\theta^*}$; the previous paragraph established such $\theta^*$ exists. Letting $\mu\coloneqq\mu_{\theta^*}$, and letting $A_\lambda\coloneqq A^*_R(\lambda\mu+(1-\lambda)\bar\mu)$ for each $\lambda\in(0,1]$, 
we will next observe $A_\lambda \subseteq (\bar a-\eta,\bar a+\eta)$ for small enough $\lambda\in(0,1]$. First, consider any $a\in A$ with $a\leq \bar a-\eta$. By definition of $\mu$ we have $\int u_S(a,\cdot)\ddd\mu < \int u_S(\bar a,\cdot)\ddd\mu$, and that $\bar a\in A^*_R(\bar\mu)$ tells us $\int u_S(a,\cdot)\ddd\mu \leq \int u_S(\bar a,\cdot)\ddd\mu$. Therefore, $a\notin A_\lambda$. Having established $A_\lambda\subseteq (\bar a-\eta,\infty)$ for every $\lambda\in(0,1]$, the desired containment for sufficiently small $\lambda$ would follow if we knew that $\limsup_{\lambda\to0} \max A_\lambda \leq \bar a$. But observe this inequality holds because $\max A^*_R$ is upper semicontinuous and $\bar a=\max A_0$.

Having seen $A_\lambda \subseteq (\bar a-\eta,\bar a+\eta)$ for sufficiently small $\lambda\in(0,1]$, and observing that $\lambda\mu+(1-\lambda)\bar\mu\in B$ for sufficiently small $\lambda\in(0,1]$, it follows that $\mu$ is as required.
\end{proof}

\begin{proof}[Proof of Theorem~\ref{ordered}]
Given Lemma~\ref{lem: heart of the ordered theorem}, the theorem follows immediately if we establish that every $\mu\in\DT$ admits $a\in\{\min A^*_R(\mu), \max A^*_R(\mu)\}$ such that $\int u_S(a,\cdot) \ddd\mu=v(\mu)$. This property follows directly from the quasiconvexity/quasiconcavity property of the ordered model: The condition clearly holds if S's expected payoff is quasiconvex in the action, and it holds vacuously (because $|A^*_R(\mu)|\leq2$) if R's expected payoff is strictly quasiconcave in the action.    
\end{proof}

\subsection{Proofs for Section~\ref{sec:pwi}}

\newcommand{\bigq}{Q}
\newcommand{\qcstrat}{q}
\newcommand{\qsstrat}{\upsilon}
\newcommand{\srange}{\mathcal{S}}
\newcommand{\joint}{\bullet}
\newcommand{\credfuns}{X}
\newcommand{\suppmarg}{\supp\;\!\!\marg}
\newcommand{\qmap}{\mathcal{Q}}
\newcommand{\subV}{\tilde V}
\newcommand{\beliefms}{\Pi}

To better understand the set of $\cred$-PBE and the S payoffs they can generate, connecting this solution concept to the analysis of \cite*{LRS1} is useful. That paper defines a notion of a \textbf{$\cred$-equilibrium} and characterizes the S payoffs such a solution can generate.\footnote{The model in LRS1 assumes the message space to be uncountable. However, given that $|M|\geq 2|\Theta|$, it follows readily from Carath\'eodory's theorem and Lemma~1 of LRS1 that the $\cred$-equilibrium payoff set is unchanged.} Roughly, a $\cred$-equilibrium specifies an experiment $\cstrat\in\cstrats$, together with continuation play and continuation beliefs that satisfy the incentive and Bayesian properties in the partial-credibility game, but with no requirement that the initial experiment $\cstrat$ be chosen optimally. From the definitions in the present paper and in LRS1, the following is immediate.
\begin{fact}\label{fact:credeq}
The quadruple $\langle\cstrat,\sstrat,\rstrat,\beliefm\rangle$ is a $\cred$-PBE if and only if: \begin{enumerate}
\item For every $\cstratr\in\cstrats$, the quadruple $\langle\cstratr,\sstrat(\cdot,\cstratr),\rstrat(\cdot,\cstratr),\beliefm(\cdot,\cstratr)\rangle$ is a $\cred$-equilibrium. 
\item We have $\cstrat \in \argmax_{\cstratr\in\cstratrs} s_\cstratr$, where $s_\cstratr$ is the S payoff induced by $\langle\cstratr,\sstrat(\cdot,\cstratr),\rstrat(\cdot,\cstratr),\beliefm(\cdot,\cstratr)\rangle$.
\end{enumerate}
In particular, every $\cred$-PBE payoff is a $\cred$-equilibrium payoff.
\end{fact}

We will use the following notation, for an S payoff that LRS1 characterizes, throughout.\footnote{The value's characterization (Theorem~1 of LRS1) remains valid in the present setting in light of Corollary~1 of LRS1.}

\begin{notation}
Let $v^*_\cred(\mu_0)$ denote the highest $\cred$-equilibrium S payoff (given prior $\mu_0$).
\end{notation}

 Toward constructing adversarial $\cred$-PBE that give S an undesirable payoff in response to off-path experiment choices, we begin with a technical lemma showing any reporting protocol comprises part of a $\cred$-equilibrium in which R always chooses from a given restricted set of best responses.

\begin{lemma}\label{lemma:approx-existence}
  If $\subV\subseteq V$ is a Kakutani correspondence and $\cstrat$ is any official reporting protocol, then some $\cred$-equilibrium $(\cstrat,\tilde\sstrat,\tilde\rstrat,\tilde\beliefm)$ exists such that $u_S(\tilde\rstrat)\in\subV(\tilde\beliefm)$.
\end{lemma}
\begin{proof}
  Let $\beliefms\coloneqq (\DT)^{M}$ be the set of all R belief mappings and define correspondences
  \begin{align*}
    \hat S \colon \beliefms &\rightrightarrows \real \\
    \tilde\beliefm & \mapsto \left[\max_{m\in M} \min \subV(\tilde\beliefm(m)), \max_{m\in M} \max \subV(\tilde\beliefm(m))\right], \\
    \hat M \colon \beliefms &\rightrightarrows M \\
    \tilde\beliefm &\mapsto \left\{m\in M\colon\, \subV(\tilde\beliefm(m)) \cap \hat S(\tilde\beliefm) \neq \varnothing\right\}.
  \end{align*}
  Observe that $\hat S$ is Kakutani, since $\subV$ is Kakutani and a finite maximum or minimum of upper or lower semicontinuous functions inherits the same semicontinuity. 
  Therefore, $\hat M$ is nonempty-valued with closed graph.
  Now let $\sstrats\coloneqq (\Delta M)^{\Theta}$ and consider the correspondence mapping belief maps into S-IC influencing strategies (assuming R's strategy delivers S values from $\subV$)
  \begin{align*}
    \hat{\sstrats} \colon\, \beliefms &\rightrightarrows \sstrats \\
    \tilde\beliefm &\mapsto  \Big\{\tilde\sstrat \in \sstrats\colon\, \cup_{\theta\in\Theta}\supp (\tilde\sstrat(\theta)) \subseteq \hat M(\tilde\beliefm)\Big\},
  \end{align*}
  and the correspondence mapping influencing strategies into consistent belief maps
  \begin{align*}
    \hat\beliefms \colon\, \sstrats &\rightrightarrows \beliefms,\\
    \tilde\sstrat &\mapsto \bigg\{\tilde\beliefm \in\beliefms\colon \ 
    \tilde\beliefm(\theta|m)\int_\Theta\bigg[\cred\ddd \cstrat(m|\cdot) + (1-\cred)\tilde\sstrat(m|\cdot) \bigg]\ddd\mu_0
     \\ & \qquad\qquad\qquad
     =
     \left[\cred(\theta)\cstrat(m|\theta) + (1-\cred(\theta))\tilde\sstrat(m|\theta)\right]\mu_0(\theta), 
    \forall\theta\in \Theta, m\in M\bigg\}.
  \end{align*}
  It then follows that $\hat{\sstrats}$ and $\hat\beliefms$ are both Kakutani.
  Therefore, the Kakutani fixed point theorem delivers some $\tilde\sstrat\in\sstrats$ and $\tilde\beliefm\in\beliefms$ such that $\tilde\sstrat\in\hat{\sstrats}(\tilde\beliefm)$ and $\tilde\beliefm\in\hat\beliefms(\tilde\sstrat)$.
  Now, take any $\infpay\in \hat S(\tilde\beliefm)$ and let $D \coloneqq \tilde\beliefm(M)$.
  Note that $\infpay \wedge\subV|_D$ is nonempty-valued and so admits a selector $\condpay\colon D \to \real$.\footnote{That is, some $\condpay\colon D \to \real$ has $\condpay\leq\infpay$ and $\condpay(\mu)\in\subV(\mu)$ for every $\mu\in D$.}
  Therefore, some $\hat\rstrat\colon D\to \Delta A$ exists such that $\us(\hat\rstrat(m)) = \condpay(m)$.
  Next, define $\tilde\rstrat\coloneqq \hat\rstrat \circ \tilde\beliefm \colon M \to \Delta(A)$.
  It is then easy to verify that $(\cstrat,\tilde\sstrat,\tilde\rstrat,\tilde\beliefm)$ is a $\cred$-equilibrium.
\end{proof}

The following lemma shows the payoff $\hat w(\mu_0)$ dominates some $\cred$-equilibrium payoff.

\begin{lemma}\label{lemma:adversarial-existence}
Every official reporting protocol $\cstrat$ admits some $\cred$-equilibrium $\langle\cstrat,\tilde\sstrat,\tilde\rstrat,\tilde\beliefm\rangle$ with ex-ante S payoff weakly below $\hat w(\mu_0)$.
\end{lemma}
\begin{proof}
Without loss, we can focus on the case that $\mu_0$ is of full support. Indeed, if we construct a $\cred$-equilibrium as desired (for official reporting protocol $\cstrat|_{\Theta_0}$) in the restricted model with state space $\Theta_0\coloneqq\supp\prior$, then this equilibrium can be extended to a $\cred$-equilibrium in the true model, by fixing any $\theta_0\in\Theta_0$ and extending $\sstrat$ to $\Theta$ via $\sstrat(\theta)\coloneqq\sstrat(\theta_0)$ for $\theta\in\Theta\setminus\Theta_0$. 
Note the lemma follows directly from Lemma~\ref{lemma:approx-existence} if we can find a Kakutani subcorrespondence $\subV\subseteq V$ such that the concave envelope of its upper selection satisfies $\mathrm{cav}[\max\subV](\mu_0) \leq \hat w(\mu_0)$.
Let us show $\subV \coloneqq [w,z]$ has this property, where $z$ is the upper semicontinuous envelope of $w$, given by
  \begin{align*}
    z \colon \DT &\to \real \\
    \mu &\mapsto \limsup_{\mu'\to\mu} w(\mu).
  \end{align*}
  First, $\subV$ is a Kakutani subcorrespondence of $V$ since $z$ is upper semicontinuous and lies above the lower semicontinuous function $w$ (and hence lies below $v$).
  All that remains, then, is to show that $\hat w(\mu_0) \geq \mathrm{cav}[\max \subV](\mu_0) = \hat z(\mu_0)$. To do so, let us establish the stronger claim that $\hat z|_D=\hat w|_D$, where $D\subseteq\DT$ is the set of full-support beliefs.
  
  Define \begin{align*}
    \tilde z \colon \DT &\to \real \\
    \mu &\mapsto \limsup_{\mu'\to\mu} \hat w(\mu).
  \end{align*}
  It follows from concavity of $\hat w$ that $\tilde z$ is concave too. Hence, 
  because $\tilde z\geq z$ and $\tilde z$ is upper semicontinuous by construction, it follows that $\tilde z\geq \hat z$.\footnote{In fact, one can show $\tilde z= \hat z$, but this fact is immaterial to the present argument.} Moreover, Theorem 10.4 from \cite{rockafellar1970convex} implies the concave function $\hat w|_D$ is continuous. Hence, $\tilde z|_D=\hat w|_D$ by the definition of $\tilde z$. That $\tilde z\geq \hat z\geq \hat w$ then implies $\hat z|_D=\hat w|_D$.
\end{proof}

Here, we provide a sufficient condition for a payoff to be compatible with $\cred$-PBE for an arbitrary credibility level.

\begin{lemma}\label{lemma:worstineq}
If $s\in \left[\hat w(\mu_0) \wedge v^*_\cred(\mu_0),\ v^*_\cred(\mu_0)\right]$, then $s$ is a $\cred$-PBE payoff for S.
\end{lemma}
\begin{proof}
  First, we argue a $\cred$-equilibrium exists with ex-ante S payoff $s$. To that end, observe that Lemma~1 from LRS1 implies $(\delta_{\mu_0},w(\mu_0),w(\mu_0))$ is a $\cred$-equilibrium outcome (as witnessed by $\poolp=\cred$ and $\goodm=\badm=\delta_{\mu_0}$). But then, as Theorem~1 from LRS1 says $v^*_\cred(\mu_0)$ is the highest $\cred$-equilibrium S payoff, it follows from Lemma~7 of LRS1 that every payoff in $[w(\mu_0), v^*_\cred(\mu_0)]$ is a $\cred$-equilibrium S payoff. Thus, $s$ is a $\cred$-equilibrium S payoff because $\hat w(\mu_0) \wedge v^*_\cred(\mu_0) \geq w(\mu_0)$. So let $(\cstrat,\tilde\sstrat,\tilde\rstrat,\tilde\beliefm)$ be some $\cred$-equilibrium generating S payoff $s$.
  
Finally, to build a $\cred$-PBE, construct $\sstrat$, $\rstrat$, and $\beliefm$ as follows. First, let $\sstrat(\cdot,\cstrat) \coloneqq \tilde\sstrat$, $\rstrat(\cdot,\cstrat) \coloneqq \tilde\rstrat$, and $\beliefm(\cdot,\cstrat) \coloneqq \tilde\beliefm$. Second, given any $\cstratr \in\cstrats\setminus\{ \cstrat\}$, 
let $\langle\cstratr,\sstrat(\cdot,\cstratr),\rstrat(\cdot,\cstratr),\beliefm(\cdot,\cstratr)\rangle$ be some $\cred$-equilibrium---which Lemma~\ref{lemma:adversarial-existence} implies exists---with ex-ante S value of at most $\hat w(\mu_0)\wedge v^*_\cred(\mu_0)$; this $\cred$-equilibrium necessarily yields S payoff no greater than $v^*_\cred(\mu_0)$ by definition of the latter.
Since $s\geq \hat w(\mu_0)\wedge v^*_\cred(\mu_0)$, the quadruple $\langle\cstrat,\sstrat,\rstrat,\beliefm\rangle$ is a $\cred$-PBE as desired.
\end{proof}

Next, we characterize the highest $\cred$-PBE payoff S can attain; it coincides with her highest $\cred$-equilibrium payoff.

\begin{lemma}\label{lem:bestpbe}
The highest $\cred$-PBE payoff for S is $v^*_\cred(\mu_0)$, her highest $\cred$-equilibrium payoff. This payoff is weakly increasing in $\cred$.
\end{lemma}
\begin{proof}
By Fact~\ref{fact:credeq}, no $\cred$-PBE payoff is strictly higher than $v^*_\cred(\mu_0)$. Meanwhile, Lemma~\ref{lemma:worstineq} implies $v^*_\cred(\mu_0)$ is a $\cred$-PBE payoff. The last statement follows from LRS1's Corollary~3.
\end{proof}

Now, we characterize the set of all $1$-PBE payoffs S can attain; it coincides with the $1$-equilibrium payoffs.

\begin{lemma}\label{lemma:allBPvalues} 
  The set of all $1$-PBE payoffs for S is $[\hat w(\mu_0),\hat v(\mu_0)]$.
\end{lemma}
\begin{proof}
Theorem~1 from LRS1 tells us $\hat v^*_1=\hat v$, and so Lemma~\ref{lemma:worstineq} says all payoffs in $[\hat w(\mu_0),\hat v(\mu_0)]$ are $1$-PBE payoffs for S. Conversely, every $1$-PBE generates a $1$-equilibrium by Fact~\ref{fact:credeq}, and so generates an equilibrium (in the sense of Definition~\ref{def:eqm}) by throwing away the influencing S behavior, the set of $1$-PBE payoffs is a subset of the set of equilibrium payoffs. Hence, the other containment follows from Proposition~\ref{perprior}.
\end{proof}

We can now prove the results on $\cred$-PBE reported in the main text.


\begin{proof}[Proof of Proposition~\ref{allcred}]

Let $S_\cred$ denote the set of $\cred$-PBE payoffs for S, and let $\bar S\coloneqq [\hat w(\mu_0), \hat v(\mu_0)]$. First, Lemma~\ref{lem:bestpbe} says $v^*(\mu_0)=\max S_\cred$, and $v^*(\mu_0)\leq\hat v(\mu_0)=\max\bar S$ by Theorem~1 from LRS1. Meanwhile,  Lemma~\ref{lemma:worstineq} implies $\bar S\cap(-\infty,\max S_\cred]\subseteq S_\cred$. Hence $S_\cred$ is weakly below $\bar S$ in the strong set order. Finally, that $S_1=\bar S$ is exactly Lemma~\ref{lemma:allBPvalues}. 
\end{proof}


\begin{proof}[Proof of Proposition~\ref{Worst}]
Fix any full-support prior $\mu_0\in\DT$, and define $\underline{s}_\cred\coloneqq w^*_\cred (\mu_0)$ for each $\cred\in[0,1]$; in particular, $\underline{s}_1=\hat w(\mu_0)$ by Lemma~\ref{lemma:allBPvalues}.
To prove the equivalence, it suffices to show that $\lim_{\cred\nearrow 1} \underline{s}_\cred = \underline{s}_1$. Further, Lemma~\ref{lemma:worstineq} tells us $\underline{s}_\cred\leq \underline{s}_1$ for every $\cred\in[0,1]$, so we need only show $\liminf_{\cred\nearrow 1} \underline{s}_\cred \geq \underline{s}_1$, which we do below. 

  Take an arbitrary $\varepsilon>0$. By definition of $\underline{s}_1$, some $p\in\RP(\mu_0)$ exists such that $\int w \ddd p > \underline{s}_1 - \varepsilon$. Moreover, by Lemma~\ref{carath}, we may further assume $|\supp(p)|\leq|\Theta|\leq|M|$.  
  For each $\mu\in \supp(p)$, let $N(\mu)\subseteq\DT$ be some open neighborhood of $\mu$ on which $w>w(\mu)-\varepsilon$, which exists because $w$ is lower semicontinuous. 
Because $\supp(p)$ is finite, which in particular implies $p(\mu)>0$ for every $\mu\in\supp(p)$, some $\underline{\cred}\in(0,1)$ is such that $$\tfrac{1}{\underline{\cred}p(\mu) + (1-\underline{\cred})}\left[\underline{\cred}p(\mu)\mu + (1-\underline{\cred}) \DT\right]
  \subseteq N(\mu)$$ for each $\mu\in\supp(p)$, and so (since $\DT$ is convex) the containment holds as well when we replace $\underline{\cred}$ with any $\cred\in(\underline{\cred},1)$.
  
  Consider now, any $\cred\in(\underline{\cred},1)$, and fix some $\cred$-PBE $\langle\cstrat,\sstrat,\rstrat,\beliefm\rangle$ generating S payoff $s\in\real$. Let $\cstratr_p\in\cstratrs$ be as defined in Proposition~\ref{perprior}'s proof, so that $p_{\cstratr_p}=p$.  
Modifying $\cstratr_p$ if necessary, we may assume without loss that any two distinct messages from $M_p\coloneqq \{m\in M\colon\, \int_{\Theta}\cstratr_p(m|\cdot)\ddd\mu_0>0\}$ would generate distinct beliefs. 
Hence, every belief $\mu\in\supp(p)$ admits
  a unique $m_\mu\in M_p$ such that every $\theta\in\Theta$ has $\tfrac{\cstratr_p(m|\theta)\mu_0(\theta)}{\int_\Theta\cstratr_p(m|\cdot)\ddd\mu_0}=\mu(\theta)$.
  If S chooses official reporting protocol $\cstratr_p$ and sends message $m_\mu$ for some $\mu\in\supp(p)$, the Bayesian property implies $\beliefm(m_\mu,\cstratr_p) \in \tfrac{1}{{\cred}p(\mu) + (1-{\cred})}\left[\underline{\cred}p(\mu)\mu + (1-\underline{\cred}) \DT\right]
  \subseteq N(\mu)$, so that R's best response property implies S has continuation value exceeding $w(\mu)-\varepsilon$. But because S chooses $\cstrat\in\cstrats$ optimally, and has the option to choose $\cstratr_p$, it must be that \begin{eqnarray*}
s &\geq& \int_\Theta \left(\int_M \left[ \int_A u_S(a) \ddd \rstrat(a|m,\cstratr_p) \right] \ddd \left[\cred\cstratr_p(m|\theta)+\left(1-\cred\right)\sstrat(m|\theta,\cstratr_p)\right] \right) \ddd \mu_0(\theta) \\
&\geq& \cred \int_\Theta \int_M \left[ \int_A u_S(a) \ddd \rstrat(a|m,\cstratr_p) \right] \ddd \cstratr_p(m|\theta) \ddd \mu_0(\theta) + (1-\cred)\min w(\DT) \\
&\geq& \cred \int_{\DT} (w-\varepsilon) \ddd p + (1-\cred)\min w(\DT) \\
&\geq& \cred (\underline{s}_1 - \varepsilon) + (1-\cred)\min w(\DT)
\end{eqnarray*}
Because $s$ was the payoff from an arbitrary $\cred$-PBE, we learn that every $\cred\in(\underline{\cred},1)$ has $\underline{s}_\cred\geq \cred (\underline{s}_1 - \varepsilon) + (1-\cred)\min w(\DT)$, which converges to $\underline{s}_1 - \varepsilon$ as $\cred$ converges to $1$. Hence, $\liminf_{\cred\nearrow 1} \underline{s}_\cred \geq \underline{s}_1-\varepsilon$.  But $\varepsilon$ was itself arbitrary, so that $\liminf_{\cred\nearrow 1} \underline{s}_\cred \geq \underline{s}_1$, as desired.
\end{proof}

\pagebreak
\bibliographystyle{jpe}
\bibliography{ref-persuasion}

\end{document}

%% file: modified-judge-pic.tex
\def\eps{0.025}
\def\thresh{0.5}
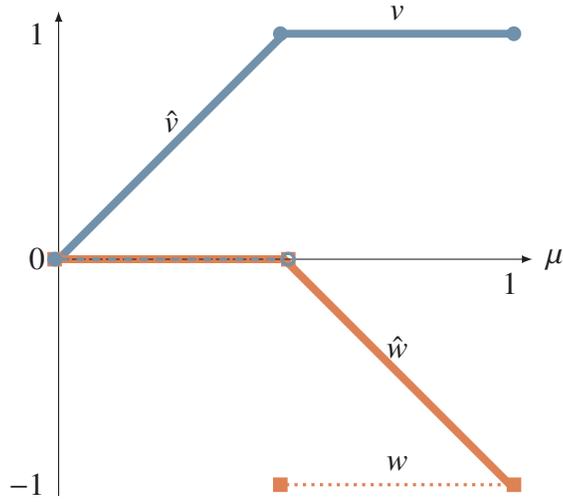
\begin{figure}[!htb]
        \centering
        \begin{tikzpicture}[xscale=6,yscale=3,font=\normalsize, inner sep = 5pt]
        \draw[blue, line width=3, opacity=0.4] (0,0) -- (\thresh,1) -- (1,1);
        \draw[red, line width=3, opacity=0.4] (0,0) -- (\thresh,0) -- (1,-1);
          \draw[latex-latex] (1.05,0) node[right] {$\mu$} -- (1,0) node[below] {$1$} -- (0,0) node[left] {$0$} -- (0,1+0.1);
          \draw (0,-1.05) -- (0,0);
          \node[left] at (0,-1) {$-1$};
          \node[left] at (0,1) {$1$};

          \draw[blue,dashed,very thick,{Circle[scale=0.9]}-{Circle[scale=0.9]}] ({\thresh-\eps},1) -- ({1+\eps},1);

          \draw[red,dotted,very thick,{Square[scale=1]}-{Square[scale=1]}] ({\thresh-\eps},-1) -- ({1+\eps},-1);
          
          \draw[red,dotted,very thick,{Square[scale=1]}-{Square[scale=1,open]}] ({0-\eps},0) -- ({\thresh+\eps},0);
          \draw[blue,dashed,very thick,{Circle[scale=0.9]}-{Circle[scale=0.9,open]}] ({0-\eps},0) -- ({\thresh+\eps},0);
          
          \node[above] at ({0.5*\thresh},0.5) {$\hat{v}$};
          \node[above] at ({0.5*\thresh+0.5},1) {$v$};
          \node[above] at ({0.5*\thresh+0.5},-0.5) {$\hat{w}$};
          \node[above] at ({0.5*\thresh+0.5},-1) {$w$};
        \end{tikzpicture}
    \caption{The functions $v, w, \hat{v}, \hat{w}$ in the modified KG's judge example.}
    \label{fig:judge-example-penalty}
\end{figure}